\documentclass[manuscript,screen]{acmart}

\usepackage[mode=buildnew]{standalone} 
\usepackage{tikz}
\usepackage{hyperref}
\usepackage{amsmath}
\usepackage{papermacros}
\usepackage[ruled, vlined, linesnumbered]{algorithm2e}
\SetArgSty{textnormal}
\SetKwComment{Comment}{$\triangleright$\ }{}
\SetCommentSty{normalfont}
\SetKw{KwDo}{do}
\SetKwFor{While}{while}{}{endwhile}
\SetKwFor{If}{if}{}{endif}
\AtBeginDocument{%
  }

\setcopyright{acmlicensed}
\copyrightyear{2018}
\acmYear{2026}
\acmJournal{TALG}
\begin{document}

\title{Adaptive encodings for small and fast compressed suffix arrays}

\author{Diego D\'iaz-Dom\'inguez}
\email{diego.diaz@helsinki.fi}
\orcid{1234-5678-9012}
\author{Veli M\"akinen}
\email{veli.makinen@helsinki.fi}
\orcid{0000-0003-4454-1493}
\affiliation{%
  \institution{Department of Computer Science, University of Helsinki}
  \city{Helsinki}
  \state{Uusimaa}
  \country{Finland}
}

\renewcommand{\shortauthors}{D\'iaz-Dom\'inguez and M\"akinen}

\begin{abstract}
Compressed suffix arrays index large repetitive collections and are widely used in bioinformatics, version control, and natural language processing. The $r$-index and its variants combine the run-length Burrows--Wheeler transform (BWT) with a sample of the suffix array to achieve space proportional to the number of equal-symbol runs. While effective for near-identical strings, their size grows quickly when edited copies are added to the collection as the number of BWT runs is sensitive to variation. Existing approaches either apply extra compression at the cost of slower queries or prioritize speed at the expense of space, limiting practical tradeoffs at large scale.

We introduce \emph{variable-length blocking} (VLB), an adaptive encoding for BWT-based compressed suffix arrays. Our technique adapts the amount of indexing information to the local run structure: compressible regions retain less auxiliary data, while incompressible areas store more. Specifically, we recursively partition the BWT until each block contains at most $w$ runs (a parameter), and organize the blocks into a tree. Accessing a position requires descending a root-to-leaf path followed by a short run scan. Compressible regions are placed near the root, while incompressible regions lie deeper and are augmented with additional indexing data to speed up access. This strategy balances space and query speed by reallocating bits saved in compressible areas to accelerate access to incompressible ones. The VLB-tree also improves memory locality by colocating related BWT and suffix array information.

Backward search relies on rank and successor queries over the BWT. We introduce a sampling technique that relaxes global correctness by guaranteeing that these queries are correct only along valid backward-search states. This restriction reduces space usage without affecting the speed or correctness of pattern matching.

We further extend the concepts of the VLB-tree to encode suffix array samples and then to the subsampled $r$-index ($sr$-index), thus producing a small fully-functional compressed suffix array. Experiments show that VLB-tree variants outperform the $r$-index and $sr$-index in query time, while retaining space close to that of the $sr$-index. The move data structure, a recent encoding, achieves better speed performance than the VLB-tree but uses considerably more space, yielding a clear space--time tradeoff.

\end{abstract}

\ccsdesc[500]{Theory of computation~Data compression}
\ccsdesc[500]{Applied computing~Computational genomics}
\ccsdesc[500]{Theory of computation~Pattern matching}

%\ccsdesc[500]{Do Not Use This Code~Generate the Correct Terms for Your Paper}
%\ccsdesc[300]{Do Not Use This Code~Generate the Correct Terms for Your Paper}
%\ccsdesc{Do Not Use This Code~Generate the Correct Terms for Your Paper}
%\ccsdesc[100]{Do Not Use This Code~Generate the Correct Terms for Your Paper}

\keywords{Burrows--Wheeler Transform, compressed suffix arrays, pattern matching, adaptive encoding, large and repetitive collections}


\maketitle

\section{Introduction}

The search for patterns in large collections of strings is a fundamental task in information retrieval, data mining, and bioinformatics. This problem has been studied for decades, producing iconic solutions such as the suffix array~\cite{ma93su} or the suffix tree~\cite{we73li}, but our ability to accumulate textual data has far exceeded the limits of those solutions, motivating the development of compressed indexes. Large collections are often repetitive, and these repetitions induce regularities in suffix arrays and trees that the compressed variants exploit to reduce space. However, the continuous accumulation of data and the development of new technologies have led to the creation of even larger collections that can reach the terabyte scale, and even compressed solutions fall short in this context.  

Compressed suffix arrays are space-efficient indexes that support the same two basic queries as plain suffix arrays: \emph{count} and \emph{locate} the occurrences of a pattern in a text. Several variants have been proposed~\cite{csagro,csasada,g2018op}, with the FM index~\cite{ferragina2005indexing} being the most widely used. The FM index represents a string $S[1..n]$ via the Burrows--Wheeler transform (BWT)~\cite{BW94} plus a sample of the suffix array $\sa[1..n]$ of $S$, and fits in $nH_k(S) + o(n\log \sigma)$ bits, where $H_k$ is the $k$-th order entropy of $S$ and $\sigma$ is the alphabet size. Given a pattern $P[1..m]$, it reports the $occ$ matches in $O(m \log \sigma + occ\log^{1+\epsilon} n)$ time for any constant $\epsilon>0$. In practice, its space is close to $n \log \sigma$ bits, far below the $n (\log \sigma + \log n)$ bits needed to support $count$ and $locate$ with $S$ and the plain suffix array. This efficiency has made the FM index a workhorse in genomics~\cite{li2010fast, lan10al}, yet even space close to the plain-text representation can be prohibitive at massive scales.

The $r$-index of Gagie et al.~\cite{g2018op} is a variant of the FM index that reduces the space to $O(r)$ by run-length encoding the BWT and sampling $2r$ elements of the SA, where $r$ is the number of equal-symbol runs in the BWT. In collections of near-identical sequences, $r$ is much smaller than $n$, achieving important space reductions that have motivated new $r$-index-based data structures~\cite{bou21ph,rossi2022moni,diaz23simp,ahmed2023spumoni,zak24movi, song24cent, tagarrays2025} for genomic analysis. However, the BWT is sensitive to edits, as a single text modification can increase $r$ by a $\Theta( \log n)$ factor~\cite{giu25bit}. The sensitivity rapidly increases $r$ with new variation until the $r$-index approaches the size of the FM index~\cite{srindex2025}. This poses a problem for terabyte-scale collections, where the level of repetitiveness is often not extreme. 

An important observation is that runs are distributed unevenly, with some BWT areas having few long runs, while others have many short runs; the suffix array samples of the $r$-index show a similar locality. This property makes uniform encoding inefficient and motivates adaptive schemes. Existing approaches only partially exploit this idea: the fixed block boost of Gog et al.~\cite{gog2019fixed} breaks the BWT into equal-length blocks and encodes them independently, while the \emph{subsampled} $r$-index (or $sr$-index) of Cobas et al.~\cite{srindex2025} discard values of $\sa$ from text areas where the sampling is dense. These techniques help but remain coarse or limited in scope.

Another challenge is query speed, as searching for long patterns in BWT-based indexes triggers multiple cache misses. The move data structure of Nishimoto and Tabei~\cite{move21} alleviates this problem by reorganizing the run-length BWT data in a cache-friendly manner, achieving significantly faster queries than the $r$-index but at the cost of being 2-2.5 times larger~\cite{bertram2024move,dinklage2025rlz}. Thus, existing work splits into two directions: minimizing space ($sr$-index) or maximizing speed (\emph{move}), with no structure achieving both.

\paragraph{Our contribution.} We present \emph{variable-length blocking} (VLB), a technique for representing BWT-based compressed suffix arrays that balances query speed and space usage. We exploit the skewed distribution of BWT runs to redistribute indexing costs. We recursively divide the BWT into variable-length blocks so that each block contains a bounded number of runs, and organize the output in a hierarchical structure (the VLB-tree). Compressible blocks remain near the root and require little auxiliary data, while less compressible regions induce deeper subtrees enriched with additional indexing information. This non-uniform layout reallocates space from areas where access is fast to areas where queries are more expensive, yielding a structure whose depth and space adapts to the compressibility of the BWT.

\begin{itemize}
\item The VLB-tree (Section~\ref{sec:vlb-tree}): we formalize the VLB-tree and analyze its performance. Given a block size $\ell$, branching factor $f$, and run threshold $w$, a query traverses a tree path of height $O(\log_f (\ell/w))$ and scans at most $w$ consecutive runs in a leaf. Counting the occurrences of a pattern $P[1..m]$ takes $O(m(\log_{f} (\ell/w) + w))$ time. Importantly, processing each query symbol incurs approximately one cache miss in sparse BWT regions and up to $\log_{f} (\ell/w)$ cache misses in dense, high-run regions, assuming $w$ consecutive runs fit the cache. 

\item Counting with toehold (Section~\ref{sec:cnt_toehold}): we integrate suffix array samples within the VLB-tree so that an occurrence of the pattern can be extracted cache-efficiently while counting. Similarly to the $r$-index, this occurrence serves as a toehold to decode the other text positions.

\item VLB-tree sampling (Section~\ref{sec:samp}): we introduce a sampling mechanism for rank and successor queries; core operations for pattern matching. The sampled data relax global correctness, only guaranteeing correct answers when the output induces a BWT range that is reachable during pattern matching. By exploiting the invariants of backward search, we reduce index space without affecting the pattern-matching speed or correctness. This relaxation principle is independent of VLB and may be applicable to other indexing schemes.

\item VLB-tree for $\phi^{-1}$ (Section~\ref{sec:phi}): we extend the VLB-tree to represent the $\phi^{-1}$ function, which decodes the occurrences of a pattern from its toehold. The combined VLB-trees of the BWT and $\phi^{-1}$ yield a compressed suffix array equivalent to the $r$-index. Section~\ref{sec:sr-vlbtree} describes how to subsample this representation to produce a cache-friendly $sr$-index, including a fast variant that speeds up the location of pattern occurrences.
\end{itemize}

We implemented two VLB-based data structures: a run-length BWT supporting $count$ queries and a fast $sr$-index supporting both $count$ and $locate$. Our experiments show speedups of up to $5.5$ for $count$ queries and $9.77$ for $locate$ queries over state-of-the-art alternatives. In terms of space, our run-length BWT achieves up to a 2-fold reduction, while our $sr$-index requires space comparable to that of Cobas et al.~\cite{srindex2025}. Although the move data structure~\cite{move21,bertram2024move} attains faster query times, it does so at the cost of being typically $2.55$--$8.33$ times larger, yielding a less favorable space–time tradeoff. Finally, cache-miss measurements indicate that the performance gains of our framework are largely due to its cache-aware hierarchical layout, underscoring its practical impact.

\section{Preliminaries}\label{sec:pre}

\paragraph{Strings.} A \emph{string} $S[1..n]$ is a sequence of $n$ symbols over an alphabet $\Sigma=\{1,2,\ldots, \sigma\}$, where the smallest symbol is a sentinel $1=\$$ that appears only at the end ($S[n]=\$$). A \emph{run} is a maximal equal-symbol substring $S[h..t]=c^{t-h+1}$, with \emph{head} $S[h]$ and \emph{tail} $S[t]$. The run-length encoding of $S[1..n] = c_1^{l_1}c_2^{l_2}{\cdots}c_{r}^{l_r}$ is $rle(S)= (c_1, l_1) (c_2, l_2) \ldots (c_r, l_r)$.

\paragraph{Operations over sequences.} \label{sec:op_seq} For $c\in\Sigma$, $rank_{c}(S, i)$ is the number of occurrences of $c$ in $S[1..i]$, and $select_{c}(S, j)$ is the position of the $j$-th occurrence of $c$ in $S$. The queries $succ_c(S, j)$ and $pred_c(S, j)$ return the nearest occurrence of $c$ at or after $j$, and at or before $j$, respectively (returning $n+1$ and $0$ if no such occurrence exists).

\subsection{The suffix array}

The \emph{suffix array} $\sa[1..n]$~\cite{ma93su} is a permutation of $[1..n]$ that lists the suffixes of $S$ in increasing lexicographic order. Its inverse $\isa[1..n]$ stores positions such that $\isa[j]=i$ iff $\sa[i]=j$. The $\sa$ supports two fundamental queries. The operation $count(P[1..m])$ returns the number of occurrences of $P \in \Sigma^{*}$ in $S$, and $locate(P[1..m])$ returns the set $\mathcal{O} \subset [n]$ of indices such that for each $i \in \mathcal{O}$, $S[i..i+m-1]=P[1..m]$.

A concept related to the suffix array is the $\phi^{-1}$ function:

\begin{definition}\label{def:phi}
\textnormal{(Definition 2 \cite{g2018op})} Let $i$ be a text position and let $\isa[i]=j$ be the corresponding index for $\sa[j]=i$. The function $\phi^{-1}$ is a permutation of $[1..n]$ defined as

\begin{equation*}
\phi^{-1}(i)= 
\begin{cases}
    \sa[j+1] = \sa[\isa[i]+1] & \text{ if } j < n \\
    \sa[1]=n & \text{ if } j = n.
 \end{cases}
\end{equation*}
\end{definition}

\subsection{The Burrows--Wheeler transform}\label{sec:bwt}

The Burrows--Wheeler transform (BWT)~\cite{BW94} of $S$ is a reversible permutation $BWT(S)=L[1..n]$ that enables compression and fast pattern search. It is defined as $L[j] = S[\sa[j]-1]$ when $\sa[j] \neq 1$, and $L[j]=\$$ when $\sa[j]=1$.

Let $F[1..n]$ be the first column of the suffixes of $S$ sorted in lexicographic order, and let $C[1..\sigma]$ be an array storing in $C[c]$ the number of symbols in $F$ that are smaller than $c$. The $LF$ mapping $LF(j) = C[L[j]] + rank_{L[j]}(L,j)$ links $L[j]$ to $F[LF(j)]=F[\isa[\sa[j]-1]]$. Moreover, visiting the $k+1$ positions $L[j], L[LF(j)=j_2], L[LF(j_2)=j_3], \ldots, L[LF(j_{k})=j_{k+1}]$ spells $S[\sa[j]-1-k..\sa[j]-1]$ from right to left. Applying $LF$ mapping $k$ times from $L[j]$ is commonly denoted $LF^{k}(j)=j_{k+1}$.

\paragraph{Backward search.} The query $backwardsearch(P)=(sp_{1}, ep_{1})$ returns the range $\sa[sp_{1}..ep_{1}]$ of suffixes in $S$ prefixed by $P[1..m]$, where $ep_{1}-sp_{1}+1 = count(P)$. In each step $l \in \{m+1, m, \ldots, 2\}$, the algorithm obtains the range $\sa[sp_{l-1}..ep_{l-1}]$ of suffixes in $S$ prefixed by $P[l-1..m]$ using the formula 

\begin{equation}\label{eq:bs}
\begin{array}{l}
   sp_{l-1} = C[P[l-1]] + rank_{P[l-1]}(L, sp_l-1)+1\\  
   ep_{l-1} = C[P[l-1]] + rank_{P[l-1]}(L, ep_l).
\end{array}
\end{equation}

The process starts with $sp_{m+1}=1$ and $ep_{m+1}=n$, which represents the empty string $P[m+1..m]$, and iteratively refines the range $\sa[sp_{l-1}..ep_{l-1}]$ until it reaches $\sa[sp_{1}..ep_{1}]$ after $m$ steps. 

%The BWT groups the symbols of $S$ with a similar right context in long runs in $L$, allowing compression proportional to the number of runs $r$. In highly repetitive datasets, $r$ is often much smaller than $n$. 

\subsection{BWT-based compressed suffix arrays}

A compressed suffix array encodes $L$ with support for queries $backwardsearch$ and $LF$ together with a sample of suffix array values. A $locate(P)$ query runs $backwardsearch(P)=(sp_{1}, ep_{1})$ and recovers the non-sampled elements of $\sa[sp_{1}..ep_{1}]$ using the links between $L$, $S$, and $\sa$. Representative designs include the FM-index~\cite{ferragina2005indexing}, the $r$-index~\cite{g2018op}, and the subsampled $r$-index~\cite{srindex2025}.

\begin{figure*}[t]
\centering
%\tikzsetnextfilename{sr_index}
\includegraphics[width=\textwidth]{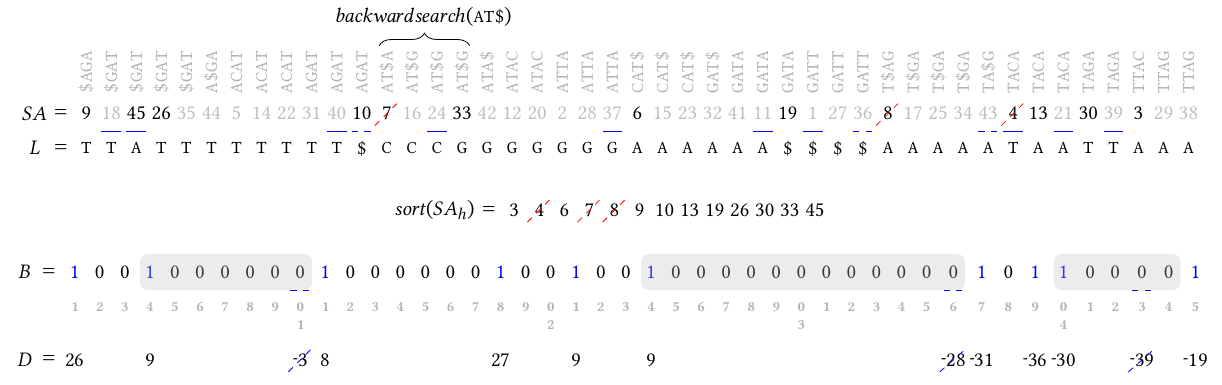}
\caption{Components of the $r$-index and $sr$-index ($s=3$) for $S=\texttt{GATTACAT\$AGATACAT\$GATACAT\$GATTAGAT\$GATTAGATA\$}$. Vertical strings are the partial suffixes of $S$ in lexicographical order. The black values in $\sa$ correspond to the sample $\sa_h$. Values of $\sa$ underlined in blue are the marked positions in $B$. The canceled values in $\sa$ were originally in $\sa_h$ but they were discarded after subsampling. The dashed underlined values in $\sa$ are the positions in $B$ that are cleared because of subsampling. Each gray region in $B$ is a partially valid area, with the underlined $0$ (cleared) marking the start of the invalid suffix. In $D$, canceled values are discarded because their bits in $B$ were cleared.}
\label{fig:sri}
\Description{Example of the $r$-index and $sr$-index}
\end{figure*}

\subsubsection{The $r$-index}\label{sec:r_index}

This data structure samples two elements of $\sa$ per BWT run and stores $L$ in a run-length encoded format that supports $LF$ mapping and backward search. The suffix array samples are encoded as follows:

\begin{itemize}
    \item An array $\sa_h[1..r]$ stores in $\sa_h[u]$ the value $\sa[h_u]$ associated with the $u$-th run $L[h_u..t_u]$.
    \item A bitvector $B[1..n]$ marks each position $B[\sa[t_u]]=1$ associated with the tail $t_u$ of $L[h_u..t_u]$.
    \item An array $D[1..r]$ stores in $D[rank_1(B[\sa[t_u]])]$ the difference $\sa[h_{u+1}]-\sa[t_u]$ between the head and tail of consecutive runs.
\end{itemize}

Figure~\ref{fig:sri} shows an example of these components. 

The $r$-index answers the query $locate(P)$ using a modified\footnote{We note that the $r$-index is often described in the literature right to left: storing run tails and using the inverse of $\phi^{-1}$. However, we defined left to right for compatibility with our method.} version of $backwardsearch(P)=(sp_{1}, ep_{1}, \sa[sp_{1}])$ that returns the first value $\sa[sp_{1}]$, and later uses $\sa[sp_{1}]$ as a \emph{toehold} to locate the elements of $\sa[sp_{1}+1..ep_{1}]$ with $\phi^{-1}$.

\paragraph{Counting with toehold.} As $backwardsearch(P)$ visits the ranges $(sp_{m+1}, ep_{m+1}), (sp_m, ep_m), \ldots, (sp_2, ep_2)$, it maintains the index $sp_g$ for the smallest step $g$ with $L[sp_g]\neq P[g-1]$, initially setting $sp_g=1$ and $g=m+1$. Once the algorithm obtains $(sp_{1}, ep_{1})$, it finds the index $\rho$ of the leftmost run $L[h_{\rho}..t_{\rho}]$ within $L[sp_g..ep_g]$ with $L[h_{\rho}] = P[g-1]$. This information is enough to compute the toehold

\begin{equation}\label{eq:toehold}
\sa[sp_{1}]= \sa_h[\rho]-(m-g+2).
\end{equation}

To retrieve $\rho$, the $r$-index uses a special successor query, say $rsucc$, that returns the index of the BWT run where the successor lies, which in this context translates to $rsucc_{P[g-1]}(L, sp_g)=\rho$.

\paragraph{Locating from toehold.} The $r$-index obtains $\sa[sp_{1}+1..ep_{1}]$ iteratively by performing $\phi^{-1}(\sa[j])=\sa[j+1]$ from $j=sp_{1}$. The difference $d=\sa[j+1]-\sa[j]$ between values $\sa[j]$ and $\sa[j+1]$ that fall within the same run $L[h_u..t_u]$, with $h_u\leq j<t_u$, is invariant under $LF$ operations as long as the traversal remains within the same run. Consequently, the difference between suffix array values changes only when crossing a run boundary.

This property allows the $r$-index to encode $\phi^{-1}$ using a single difference per run. Let $b$ be the largest index such that $B[b]=1$ and $b \leq \sa[j]$. Then,

\begin{equation}\label{eq:phi}
\phi^{-1}(\sa[j]) = \sa[j]+D[rank_1(B, \sa[j])] = \sa[j+1].
\end{equation}

Thus, $\phi^{-1}$ reduces to a predecessor query on $B$ followed by a lookup in $D$.

\subsubsection{The subsampled $r$-index}\label{sec:sr_index}

The $sr$-index reduces space by discarding suffix array samples in dense regions of $\sa_h$. Removed samples invalidate some regions of the index, which are reconstructed at query time via $LF$ operations. This approach significantly reduces space while preserving good query performance.

Let $p_1 < p_2 < \ldots < p_r$ be the sorted values of $\sa_h$, and let $s<n$ be a sampling parameter. For any $1<a<r$, the sample $p_a$ is discarded from $\sa_h$ if $p_{a+1}-p_{a-1} \leq s$. The test is applied sequentially for $a=2,3,\ldots, r-1$. Thus, if $p_a$ is removed, the next step checks if $p_{a+2}-p_{a-1} \leq s$ to decide whether $p_{a+1}$ is sampled, applying the same procedure to the remaining elements of $\sa_h$. 

A bitvector $R[1..r]$ now marks the runs in $L$ that were subsampled. A rank structure maps the retained samples to their positions in the reduced array $\sa_h$. Moreover, for each discarded run $L[h_u..t_u]$, the associated bit $B[\sa[h_{u}-1]]=1$ is cleared and the corresponding difference $D[rank_1(B, \sa[h_{u}-1])]$ removed.

\paragraph{Toehold recovery} During the execution of $backwardsearch(P)=(sp_{1}, ep_{1}, \sa[sp_{1}])$, the run $\rho$ that determines the toehold $\sa[sp_{1}]$ (Equation~\ref{eq:toehold}) may no longer retain a sample. In this case, the $sr$-index recovers $\sa[h_{\rho}]$ by following $LF$ steps until reaching a run head $h_{\rho'}=LF(h_{\rho})^{k}$, with $k\leq s$, whose sample is retained ($R[\rho']=0$). The value is then recovered as

\begin{equation}\label{eq:sr_index_th}
\sa[h_{\rho}]=\sa_h[rank_0(R, \rho')]+k. 
\end{equation}

\paragraph{Valid and invalid areas} Subsampling also affects $\phi^{-1}$. Let $b$ be the largest position such that $B[b]=1$ and $b<\sa[j]$. Equation~\ref{eq:phi} is correct iff no 1-bits were removed from the interval $B[b..\sa[j]]$. Otherwise, $\sa[j]$ lies in a partially valid area.

In the standard $sr$-index, the validity is detected by performing up to $s$ $LF$ steps starting from $j+1$. If a run head with a retained suffix array sample is encountered, the value of $\phi^{-1}(\sa[j])$ is recovered from that sample; otherwise, $B[b..\sa[j]]$ is guaranteed to be valid and Equation~\ref{eq:phi} applies.

\subsubsection{Fast $sr$-index variant}\label{sec:sr-index-fast}

The standard $sr$-index always performs up to $s$ $LF$ steps to test validity, even when Equation~\ref{eq:phi} applies. The fast variant (Cobas et al.~\cite{srindex2025}) avoids this overhead by explicitly marking invalid areas.

Let $b_a$ and $b_{a+1}$ be consecutive 1-positions in the subsampled $B$. A bitvector $V$ stores $V[a]=1$ if the interval $B[b_a..b_{a+1}-1]$ is fully valid, and $V[a]=0$ if it is partially valid. In the latter case, a prefix remains valid, while the complementary suffix is an invalid area. An array $M$ stores the lengths of valid prefixes in partially valid intervals. Let $b_a<b'<b_{a+1}$ be the leftmost cleared position in a partially invalid $B[b_a..b_{a+1}-1]$ ($V[a]=0$). The entry $M[rank_0(V, a)]$ stores $b'-b_{a}$.

Given a $\sa[j] \in [b_a, b_{a+1}-1]$, the operation $\phi^{-1}(\sa[j])$ is computed as follows:

\begin{enumerate}
    \item If $V[a]=1$, the area is valid and Equation~\ref{eq:phi} is correct.
    \item If $V[a]=0$ and $b+M[rank_0(V, a)]>\sa[j]$, $\sa[j]$ falls within a valid prefix and Equation~\ref{eq:phi} still applies.
    \item If $V[a]=0$ and $b+M[rank_0(V, a)]\leq \sa[j]$, then $\sa[j]$ falls within an invalid area and $\phi(\sa[j])$ is recovered via $LF$ steps.
\end{enumerate}

This scheme ensures that $LF$ traversal is only used when strictly necessary, significantly reducing query time in practice.

Figure~\ref{fig:sri} illustrates the components of the $sr$-index and its fast variant. 

\section{The VLB-tree: an adaptive representation of the BWT}\label{sec:vlb-tree}

A single random edit in the text typically creates multiple new short runs in the BWT. This happens because even a local modification may change the lexicographic order of many suffixes, which ultimately determines the ordering of symbols in the BWT. As a consequence, small edits may have a disproportionate impact on the number of BWT runs $r$, directly increasing the space usage of the $r$-index, whose size depends on that number. 

A key observation, however, is that these newly created short runs are not distributed uniformly across the BWT. In practice, some regions remain highly compressible after edits, while others become locally irregular. This phenomenon is also reflected in the BWT's ability to benefit from block-based compression, where independent zeroth-order encoding of blocks can still achieve high-order compression~\cite{ferragina2005indexing}.

The $r$-index does not fully exploit this property, as it treats all runs uniformly, regardless of their local density or distribution. This motivates an adaptive representation that varies its granularity according to the local distribution of runs, keeping compressible regions compact while recursively refining irregular ones.

We now introduce \emph{variable-length blocking} (VLB) to form a hierarchical representation of the BWT that realizes this adaptive strategy (the VLB-tree). Our encoding naturally enables a more cache-friendly layout, as large compressible regions can be processed with minimal memory accesses, while irregular regions are explored only when necessary.

\subsection{Overview of the VLB-tree}\label{sec:basic_end}

The VLB-tree decomposes $L$ into blocks whose sizes adapt to the local distribution of runs:

\begin{itemize}
    \item Compressible regions (few runs) are stored in large leaf blocks.
    \item Irregular regions (many runs) are subdivided recursively until their fragments have a manageable number of runs.
\end{itemize}

We begin by dividing $L$ into $n/\ell$ blocks $L_1,L_2,\ldots, L_{n/\ell}$ of size $\ell$, processed from left to right. For each $L_b$, we compute $rle(L_b)$ and compare its number of runs to a threshold $w$, the maximum allowed in a leaf.

\paragraph{Compressible blocks} If $rle(L_b)$ has fewer than $w$ runs, we try to extend it by concatenating consecutive blocks $L_{b}{\cdots}L_{b+z-1}$. We increase $z$ while $rle(L_{b,b+z-1})$ has at most $w$ runs and stop when $rle(L_{b,b+z})$ exceeds this threshold. The fragment $rle(L_{b,b+z-1})$ becomes a leaf superblock whose parent is the root of the tree. 

\paragraph{Incompressible blocks} If $rle(L_{b})$ has more than $w$ runs, $L_b$ is considered incompressible and is subdivided. Given a branching factor $f$, we split $L_{b}$ into $f$ equal-length parts of size $|L_b|/f$ and process them recursively from left to right using the same merging-and-splitting rule. Some parts may merge into a leaf superblock; others may subdivide further. The process adapts automatically to the local run structure.

The resulting structure places highly compressible regions close to the root, yielding small access costs, while irregular regions are recursively refined and incur access costs that reflect their local run complexity.

%The recursive partitioning by a constant factor $f$ simplifies access to the BWT. Given a position $L[i]$, find the block $b=\lceil i/\ell \rceil$, the block position $i=i - (b-1)\ell$, and apply the idea recursively with block size $\ell/f$ until a leaf is reached.

Recursive partitioning by a constant factor $f$ simplifies navigation to the appropriate fragment of the BWT. Given a position $L[i]$, we compute the top-level block $b=\lceil i/\ell \rceil$, its local position $i'=i-(b-1)\ell$, and apply the same procedure recursively with block size $\ell/f$ until a leaf is reached.

At the same time, the formation of superblocks reduces the number of nodes and thus the space overhead. Data removed by merging blocks is recomputed on the fly by scanning up to $w$ runs in a cache-friendly manner, which remains practical when $w$ is small. The only problem with superblocks is that they break direct index-based access to the BWT. We restore efficient navigation by storing compact bitvectors that map positions recursively to the correct child blocks.

\subsection{Layout and access in the VLB-tree}

Leaves of the VLB-tree store BWT fragments, while internal nodes store routing information that enables efficient navigation to these fragments.

Let $L^{u}$ denote the recursive fragment of $L$ associated with a node $u$. If $u$ is not the root, there exists a parent node $v$ whose fragment $L^{v}$ was partitioned into equal-length blocks, and where one of its blocks $L^{v}_b$ corresponds to the start of $L^{u}$.

Let $\Sigma^{v}=[1..\sigma^{v}]$ be the alphabet of $L^{v}$, and let $\Sigma^{u} \subseteq \Sigma^{v}$ be the subset of symbols appearing in $L^{u}$. The operator $alphacomp(L^{u})$ maps the symbols of $\Sigma^{u}$ to the compact alphabet $[1..\sigma^{u}]$. Unless stated otherwise, symbols in $L^{u}$ are referred to using their compacted identifiers.

Each node stores its local alphabet, which typically becomes smaller at deeper levels of the tree. Several auxiliary structures are defined with respect to $\Sigma^{u}$; the first is a prefix-count array used to answer rank queries.

Thus, the first two components stored at node $u$ are:

\begin{itemize}
    \item A bitvector $A[1..\sigma^v]$ that sets $A[c]=1$ if the parent symbol $c \in \Sigma^{v}$ appears in $L^{u}$.
    \item  A prefix-count array $O[1..\sigma^u]$ where the symbols of $\Sigma^{u}$ are indexed in compacted form. For each parent symbol $c \in \Sigma^{v}$ with $A[c]=1$, $O[rank_1(A, c)]$ stores the number of occurrences of $c$ in $L^{v}_{1,b-1}$.
\end{itemize}

The remaining components depend on whether $u$ is a leaf or an internal node.

\subsubsection{Leaves}

A leaf $u$ corresponds to a maximal fragment $L^u = L^v_b \cdots L^v_{b+z-1}$, with $z \geq 1$, such that $rle(L^u)$ contains at most $w$ runs, but including block $L^v_{b+z}$ would exceed $w$ runs. The leaf stores:

\begin{itemize}
    \item An array $E = rle(alphacomp(L^{u}))$.
\end{itemize}

\paragraph{Encoding of the runs}
We encode the runs in $E$ using a space-efficient layout that supports fast scans. Each run is stored as an integer $x$, whose first $p$ bits encode the run symbol and whose remaining bits encode the run length, where $p$ depends on the local alphabet of $E$. The encoding applies to these integers. Our design is inspired by the vbyte-based scheme of Lemire and Boytsov~\cite{fastvbytes} and enables parallel decompression, making it amenable to SIMD implementations.

We represent $E$ as a byte stream $C_1U_1 \cdot C_2U_2 \cdots C_kU_k$, where each 8-bit control byte $C_k$ specifies the byte lengths of the runs stored in the variable-length stream $U_k$. The number of runs packed in each $U_k$ depends on the largest run encoding in the leaf: eight runs if all fit in at most two bytes, four runs if up to four bytes are needed, and two runs otherwise (up to eight bytes). Accordingly, the bits of $C_k$ encode the byte length of each run.

The number of runs packed in each $C_kU_k$ block also determines the degree of parallel decompression. As a special case, when all runs fit in a single byte, we store $E$ in plain format. During VLB-tree queries, SIMD instructions are used to accelerate scans of $E$.

\subsubsection{Internal nodes}\label{sec:int_node_enc}

A node $u$ is internal when its fragment $L^{u}=L^{v}_b$ contains more than $w$ runs. In this case, we divide $L^{u}$ into $f$ equal-length parts $L^{u}_1, L^{u}_2, \ldots, L^{u}_f$ of length $\ell_u = |L^{u}|/f$, and process them recursively from left to right using the same merging-and-splitting rule.

Some parts become internal nodes (if they have $>w$ runs), while other consecutive parts within the same parent may merge into leaf superblocks (if their concatenation has $\leq w$ runs).  

The variable-length partition of $L^{u}$ prevents direct index-based access, but we restore efficient navigation with additional routing information. Let $L^{u}[i]$ be a position falling within a fragment $L^{e}=L^{u}_{o,o+z-1}$ of a child $e$ of $u$, with $z \geq 1$. We compute:

\begin{itemize}
    \item The index $b=\lceil i/\ell_u \rceil$ for the block $L^{u}_{b}$ where $i$ falls, with $b \in [o..o+z-1]$.
    \item The index $o$ for $L^{u}_{o,o+z-1}$. This value maps $L^{u}[i]$ to its position $i-(o-1)\ell_u$ within $L^{e}$.
    \item The child index $q$ for $e$ among the children of $u$. This information tells us what pointer to follow to visit $L^{e}$.
\end{itemize}

Accordingly, an internal node $u$ stores:

\begin{itemize}
    \item A bitvector $X[1..f]$ marking which of the $f$ parts begin a child. If $f=4$ and $X=1001$, then $L^{u}_{1},L^{u}_2,L^{u}_3$ defines the first child and $L^{u}_{4}$ the second, so $u$ has $f'=2$ children.   
    \item An array $P[1..f']$ containing pointers to the children of $u$.
\end{itemize}

We compute $o=pred_1(u.X, b)$, $q=rank_1(u.X, o)$, and follow the pointer $u.P[q]$ to reach $e$.

The branching factor $f$ is kept small (e.g. $f \leq 64$) so that $pred_1$ and $rank_1$ in $X[1..f]$ are implemented with bitwise operations on a single machine word. Figure~\ref{fig:simp_vlbt_enc} and Example~\ref{ex:simp_vlbt_enc} illustrate the VLB-tree construction.

\begin{figure*}[t]
\centering
%\tikzsetnextfilename{vlb-tree}
\includegraphics[width=0.7\linewidth]{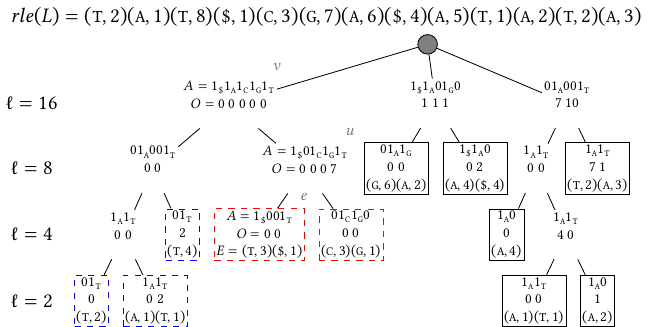}
\caption{VLB-tree built on the BWT $L$ from Figure~\ref{fig:sri}. The parameters are $\ell=16$, $f=2$, and $w=2$. In nodes $v, u, e$, we label the stored node fields (omitted elsewhere for readability). We also omit array $X$ as it is the same in all nodes ($X=11$, no superblocks), and the pointer array $P$. In the alphabet bitvector $A$, a $1$ indicates that the corresponding parent symbol appears in the node. Subscripts annotate each $1$ (and the run-length sequence $E$) with the original text symbols. Dashed boxes correspond to Example~\ref{ex:simp_vlbt_enc}.} 
\label{fig:simp_vlbt_enc}
\Description{Example of the VLB-tree}
\end{figure*}

\begin{example}\label{ex:simp_vlbt_enc}
The VLB-tree of Figure~\ref{fig:simp_vlbt_enc}. The sequence of $L$ is initially partitioned into $\lceil n/\ell \rceil= \lceil 45/16 \rceil = 3$ blocks $L_1,L_2,L_3$. The first block $rle(L_1)=(\texttt{T},2)(\texttt{A},1)(\texttt{T},8)(\texttt{\$},1) (\texttt{C},3)(\texttt{G},1)$ has $6>w$ runs, so we create a new internal node $v$. The local alphabet of $L_1$ has all the symbols of $\Sigma$ and $L_1$ does not have a left context, so the alphabet bitvector is $v.A=11111$ and the prefix-count array is $v.O=0 0 0 0 0$. We now build the subtree of $v$ using $L^{v}=alphacomp(L_1)$ and block size $\ell=16/2=8$. $L^v$ has two parts $L^{v}_1,L^{v}_2$, and their compressed sequences are $rle(L^{u}_1)=(\texttt{T},2)(\texttt{A},1)(\texttt{T},5)$ (blue dashed boxes) and $rle(L^{u}_2)=(\texttt{T},3)(\texttt{\$},1) (\texttt{C},3)(\texttt{G},1)$ (red dashed boxes). Both yield internal nodes as they have more than $w$ runs each. Assume that we have already processed $L^{v}_1$, and now we have to create an internal node $u$ for $L^{u} = alphacomp(L^{v}_2)$. The local alphabet of $L^{v}_2$ is $\Sigma^{u} = \{\$,\texttt{C},\texttt{G},\texttt{T}\}$, so the alphabet bitvector becomes $u.A=10111$. From the symbols in $L^v_2$, only $\texttt{T}$ (compact index $4$) occurs in the left context $L^{v}_{1}$ (seven times), meaning that the prefix-count array is $u.O=0007$. The block size now becomes $\ell=8/2=4$ and the rest of the tree is computed similarly.
\end{example}

\subsubsection{The root as special case}\label{sec:root_enc}

The VLB-tree root represents the entire sequence $L$, and therefore does not store the prefix-count array. However, unlike internal nodes\emdash which have at most $f$ children\emdash the root may have up to $\lceil n/\ell \rceil$ children. Storing the bitvector $X[1..\lceil n/\ell \rceil]$ with fast $rank_1$ and $select_1$ operations would be costly.

Instead, we use an array $P[1..\lceil n/\ell \rceil]$ that encodes both the children's pointers and collapsed runs of zeros in $X$. If the (logical) bitvector $X$ has $X[b]=1$, then $P[b]$ stores a pointer to child $q=rank_1(X,b)$ of the root. If $X[b]=0$, then $P[b]$ stores the offset $b-pred_1(X, b)$. The least significant bit of $P[b]$ indicates whether it is a pointer or an offset.

To locate the root child covering $L[i]$, we compute the top-level block $b = \lceil i/\ell \rceil$, read $P[b]$, and either follow the pointer directly (if it is a pointer), or follow the pointer stored at $P[b-P[b]]$ (if $P[b]$ is an offset).

Although $P$ may be larger than storing $X$ plus rank and predecessor data structure, it keeps navigation cache-friendly.

\begin{algorithm}[t]
\DontPrintSemicolon
\caption{$access(L, i)$}\label{algo:access}

$u \gets \text{root of the VLB-tree}$\;
$b \gets \lceil i/\ell \rceil, o \gets b$\;
\lIf{$P[b]$ is an offset}{$o \gets b-P[b]$}

$i \gets i - (o-1)\ell$\;
$\ell \gets \ell/f$\;
$u \gets \text{node in } u.P[o]$\;
${P} \gets \emptyset$\;

\While{$u \neq \text{leaf}$}{
$P \gets P \cup \{u\}$\;
$b \gets \lceil i/\ell \rceil, o \gets pred_1(u.X, b), q \gets rank_1(u.X, o)$\;
$i \gets i - (o-1)\ell$\;
$\ell \gets \ell/f$\;
$u \gets \text{node in } u.P[q]$\;
}

$c \gets \text{symbol of the run } u.E[z] \text{ that contains } i$\;
$c \gets select_1(u.A, c)$\;

\For{$j \gets |P|$ \textbf{to} $1$}{
$u \gets P[j]$\;
$c \gets select_1(u.A, c)$\;
}

\Return {$c$}\;
\end{algorithm}

\paragraph{Practical considerations.}
The construction of the tree uses the parameters $\ell$, $f$, and $w$. We assume $\ell > w$ and that both $\ell$ and $w$ are powers of $f$; that is, $\ell = f^{x}$ and $w = f^{y}$ for some integers $x$ and $y$. All nodes are stored in a contiguous memory region in depth-first (DFS) order. Each node’s encoding is byte-aligned, so the pointers in the array $P$ are byte offsets into this memory region.

\subsection{Accessing positions of the BWT}\label{sec:acc}

For any position $i$, define $path(i) = u_1, u_2,\ldots u_l$ as the list of nodes in the VLB-tree visited to locate $L[i]$, where $u_1$ is the root and $u_l$ is a leaf.

The operation $access(L, i)$ consists of three phases: a top-down descent to identify the leaf containing position $L[i]$, a local scan of the leaf to find the compacted symbol, and a bottom-up ascent to recover the original symbol in $\Sigma$. 

The descent follows the navigation mechanism described in Sections~\ref{sec:int_node_enc} and~\ref{sec:root_enc}. Starting from the root, we repeatedly use the routing information stored at each internal node to select the child fragment containing the current position and update $i$ to its local offset within that fragment. The descent terminates at a leaf $u_l$ such that $L[i] \in L^{u_l}$.

When the descent halts at the leaf $u_l$, we scan $u_{l}.E$ to find the leftmost run $u_l.E[z]=(c, l)$ whose cumulative length reaches $i$.

We then ascend along $u_l, u_{l-1},\ldots,u_2$ updating the compacted symbol $c$ at each level to the corresponding parent symbol $c=select_{1}(u_j.A,c)$. This process yields the original $L[i]$ in $\Sigma$.

Algorithm~\ref{algo:access} shows the process in more detail.

\begin{theorem}
Let $\mathcal{T}$ be a VLB-tree built on top of the BWT $L$ of $S$ with parameters $f, w$, and $\ell$. Accessing $L[i]$ in $\mathcal{T}$ takes $O(\log_f (\ell/w) + w)$ time.
\end{theorem}

\begin{proof}
The total time is the height of $\mathcal{T}$ plus the time spent scanning a leaf. During the construction of $\mathcal{T}$, blocks of size $\ell$ are repeatedly divided by a factor $f$ until the resulting fragments have at most $w$ runs. We choose $\ell=f^{x}$ and $w=f^{y}$ with $x>y$. A node at depth $h$ encodes a fragment of length $\ell/f^{h-1}$. These fragments range from size $\ell$ (children of the root) down to $w$ (where the subdivision stops). Thus the maximum depth is $\log_f \ell - \log_f w = \log_f (\ell/w)$ nodes. Each leaf has $\leq w$ runs, which gives the additive term $w$.
\end{proof}

The worst case arises when roughly $\ell$ BWT runs of length~1 occur consecutively\emdash an uncommon situation in highly repetitive data. If $i$ falls inside a leaf superblock attached directly to the root\emdash a common case in repetitive inputs\emdash, then $path(i)$ has two nodes, so the time drops to $O(w)$. Choosing a small threshold $w$ (e.g., $\approx 64$) makes the final linear scan over $E$ effectively constant using SIMD. The number of cache misses along $path(i)$ is bounded by $\log_f(\ell/w)$ (assuming each node fits in cache), and deeper nodes exhibit better locality. A small $w$ also keeps $E$ cache-resident, further reducing misses.

\subsection{Adding support for rank and successor}\label{sec:aug_enc}

The compressed suffix array requires rank and successor queries to support pattern matching. However, implementing these operations in the VLB-tree is challenging because each node only stores information about its local alphabet, whereas the queries receive as input an arbitrary symbol $c$ and a position $i$. We augment the VLB-tree to address this issue. 

The operations $rank_c(L, i)$ and $succ_c(L, i)$ return different information, but they can be implemented using the same traversal strategy in the VLB-tree. We begin by following $path(i)=u_1,u_2 ,\ldots, u_l$ and stop at the first node whose local alphabet does not contain $c$. At that level, the traversal diverges to a close node (typically a sibling) that leads to the answer. In the case of the rank query, this corresponds to the one leading to the closest occurrence of $c$ to the left of $i$ ($pred_c(L, i)$), while for the successor query, it leads to the closest occurrence of $c$ to the right of $i$. 

Our current encoding does not allow us to infer the diverging nodes efficiently, so we will augment the internal nodes with additional routing information to visit these paths.

Consider an internal node $u$ corresponding to the fragment $L^{u}$. Its compacted alphabet is $\Sigma^{u} = [1..\sigma^{u}]$, and the partition of $L^{u}$ induces $f^u$ children for $u$. We add:

\begin{itemize}
    \item A bitvector $N[1..f^u\sigma^{u}]$ storing the concatenation of $\sigma^{u}$ bitvectors $N_{1}[1..f^u]{\cdot}N_{2}[1..f^u]{\cdots}N_{\sigma^{u}}[1..f^u]$. For a symbol $c \in \Sigma^{u}$ and the child index $q \in [1..f^{u}]$, $N_c[q]=1$ iff the $q$-th child of $u$ contains $c \in \Sigma^{u}$.
\end{itemize}

The bitvector $N$ indicates whether the symbol $c$ appears in each child of $u$, and using $succ_1$ and $pred_1$ queries on $N_c$ redirects the traversal to the nearest left/right sibling containing $c$. When $N_c$ fits a machine word, these queries run in $O(1)$ time via word operations. The additional $f^{u}\sigma^{u}$ bits per internal node impose only a small overhead because $f^{u} \leq f$ and the local alphabet size $\sigma^{u}$ is small in repetitive texts and typically decreases at deeper tree levels.

The root also requires such routing information. However, the size of its corresponding bitvector can reach $\sigma(n/\ell)$ bits, too large to fit a machine word, thus requiring auxiliary data structures to support $succ_1$ and $pred_1$. These extra structures would likely reside far from the root in memory, causing cache misses\emdash similarly to the issue with bitvector $X$ (Section~\ref{sec:basic_end}).

To avoid this overhead, we reorganize the routing information of the root by distributing it across its children. 

For each root child $u$, we create:

\begin{itemize}
    \item An array $Z[1..\sigma]$ that stores in $Z[c]$ the distance between $u$ and its closest right-hand sibling $y$ such that $c$ occurs in $L^{y}$. Thus, if $u$ and $y$ are the children $q$ and $q'$ of the root (respectively), then $u.Z[c]=q'-q$. 
\end{itemize}

We use $Z$ only when redirection is required from the second node (root child) in the traversal, ensuring that the necessary routing information is local and avoiding cache misses.

The aggregate space of all $Z$ arrays may exceed that of the root's bitvector $N$ and its auxiliary data structures, but we adopt this design to preserve cache efficiency. In Section~\ref{sec:samp}, we describe a technique to reduce the space of $Z$ with no penalty in query time.

\begin{figure}[t]
\centering
\begin{minipage}[t]{0.49\textwidth}
\begin{algorithm}[H]
\DontPrintSemicolon
\caption{$rank_c(L, i)$}\label{algo:rank_q}
$u_2 \gets \text{ root child in } path(i) $\;
\If{$u_2.A[c]=0$}{
$y_1 \gets \text{node indicated by } u_2.Z[c]$\;
\Return{$y_1.O[rank_1(y_1.A, c)]$}
}

$r \gets 0$\;

\For{$u_j \in u_2, u_3, \ldots, u_l \subset path(i)$}{
    $c \gets rank_1(u_j.A, c)$\;
    $r \gets r + u_j.O[c]$\;
    $q \gets \text{ child index of } u_{j+1} \text{ in } u_j$\;
    \If{$u_j.N_c[q]=0$}{
        $y_1 \gets \text{nearest left sibling of } u_{j+1} \text{ with } c$\;
        \lIf{$y_1=\textsc{NULL} \text{\textbf{ then}}$}{\Return{$r$}}
        \For{$y \in y_1, y_2, \ldots, y_l$}{
            $c \gets rank_1(y.A, c)$\;
            $r \gets r + y.O[c]$\;
        }
        \Return {$r + \text{cumulative run lengths of } c \text{ in } y_l.E$}\;
    }
}
$r' \gets \text{cum. run lengths of } c \text{ in } u_l.E \text{ up to local } i$\; 
\Return {$r + r'$}\;
\end{algorithm}
\end{minipage}
\begin{minipage}[t]{0.49\textwidth}
\begin{algorithm}[H]
\DontPrintSemicolon
\caption{$succ_c(L, i)$}\label{algo:succ_q}

$u_2 \gets \text{ root child in } path(i) $\;

\If{$u_2.A[c]=0$}{
$y_1 \gets \text{node indicated by } u_2.Z[c]$\;
\For{$y \in y_1, y_2, \ldots, y_l$}{\label{code:descend_start}
    $c \gets rank_1(y.A, c)$\;
}
\Return{$\text{global idx of the first } c \text{ in } y_l.E[z]$}\label{code:descend_end}
}

$c_k \gets c, p_k \gets u_2.Z[c]\text{ }$\Comment{candidate for divergence}

\For{$u_j \in u_2, u_3, \ldots, u_l \subset path(i)$}{

    $q \gets \text{ child index of } u_{j+1} \text{ in } u_j$\;
    $q' \gets succ_1(u_{j}.N_c, q+1)\text{ }$\Comment{next sib. of $u_j$ with $c$}
    \lIf{$q' \neq \textsc{NULL}$ \text{\textbf{then}}}{ $c_k \gets c$, $p_k \gets u_j.P[q']$ }
    
    \If{$u_j.N_c[q]=0$}{
        $c \gets c_k$, $y_1 \gets \text{ node indicated by } p_k$\;
        execute Lines~\ref{code:descend_start}-\ref{code:descend_end}\;
    }
    $c \gets rank_1(u_j.A, c)$\;
}

$\text{find the first run } u_l.E[z] \text{ for } c \text{ after local } i \text{ in } u_l$\;

\lIf{$u_l.E[z]$ \text{is found \textbf{then}}}{\Return{global idx its first $c$}}

$c \gets c_k$, $y \gets \text{node indicated by } p_k$\;
execute Lines~\ref{code:descend_start}-\ref{code:descend_end}\;
\end{algorithm}
\end{minipage}
\Description{Pseudocode for $rank$ and $succ$ queries}
\end{figure}

\subsubsection{Rank algorithm}\label{sec:rank_q}

We begin the algorithm for $rank_c(L, i)$ by descending $path(i) = u_1, u_2, \ldots, u_l$ as described in Section~\ref{sec:acc}. Throughout the process, the variable $c$ is updated to denote the compacted symbol corresponding to the current node.

\paragraph{Root-child special case} If $u_2.A[c]=0$, the fragment $L^{u_2}$ does not contain $c$, so we visit the right-hand sibling $y_1$ indicated by $u_2.Z[c]$ and report the prefix count from $y_1.O[c]$. Node $y_1$ lies on $path(succ_{c}(L, i))$, not on  $path(pred_c(L, i))$, which we previously stated would be used for $rank_c$. However, the answer in $y_1.O$ is still correct because it considers the occurrences of $c$ preceding $L^{u_2}$. This is exceptional because $u_2.Z$ does not provide information for redirecting to a left-hand sibling.

\paragraph{General internal node} This case assumes $u_2.A[c]=1$. For an internal node $u_{j} \neq u_{2}$ with child index $q$ for $u_{j+1}$, if $u_{j}.N_c[q]=0$, then $u_{j+1}$ does not contain $c$. We therefore compute $y_1 = pred_1(u_j.N_c, q-1)$, the closest left-hand sibling of $u_{j+1}$ that contains $c$. If there is no $y_1$, we return the cumulative prefix count for $c$ collected up to $u_j$. If $y_1$ exists, we descend the rightmost branch $y_1,y_2, \ldots, y_l$ that contains $c$, computing each successive child as $y_{j+1}= pred_1(y_{j}.N_{c}, f^{y_j})$, where $f^{y_j}$ is the number of children for $y_j$. Upon reaching the leaf $y_l$, we scan its entire sequence $y_l.E$, count all occurrences of $c$, and add them to the cumulative rank for $c$ along the nodes $u_2,\ldots,u_j$ and the redirected branch $y_1,\ldots y_l$. The full leaf scan is necessary since the subtree of $y_1$ lies entirely before position $i$, so all occurrences in $y_l$ contribute to rank. The final rank value is the answer for $rank_{c}(L, i)$. 

Algorithm~\ref{algo:rank_q} illustrates the procedure.

\subsubsection{Successor algorithm}\label{sec:succ_q} 

As mentioned above, $succ_c(L, i)$ operates almost the same way as $rank_c(L, i)$. The most notable change is the node we search after diverging. The query returns a null value if no valid successor for $c$ exists in $L[i..n]$. 

We descend $path(i)$, compacting $c$ accordingly, until we reach the first node $u_{j}$ where $u_{j+1}$ does not contain $c$ ($u_j.N_c[q]=0$). From there, we locate the deepest ancestor $u_k \in path(i)$, with $k \leq j$, such that $L^{u_{k}}$ has an occurrence of $c$ to the right of the local $i$.

\begin{figure}[t]
\centering
%\tikzsetnextfilename{succ_rank_img}
\includegraphics[width=0.25\linewidth]{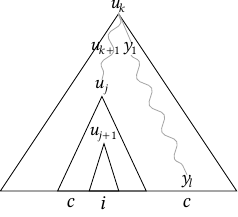}
\caption{Finding the successor node containing $c$.} 
\Description{Finding the successor node}
\label{fig:succ_q}
\end{figure}

Once we locate $u_{k}$, we move to the closest right-hand sibling $y_1$ of $u_{k+1}$, and continue descending the leftmost branch $y_1,y_2,\ldots,y_l$ that contains $c$. The initial compacted symbol is $c_{k}$ and each subsequent node $y_{j+1}$ is the leftmost child of $y_{j}$ that contains $c$. Upon reaching $y_l$, the algorithm scans $y_{l}.E$ until it finds the first run $y_l.E[z]$ for $c$, and reports the index in $L$ for the first $c$ in $y_l.E[z]$ (i.e., its global index). Figure~\ref{fig:succ_q} illustrates the divergence mechanism.

If no stopping node $u_j$ is found and the descent of $path(i)$ reaches its leaf $u_{l}$, the algorithm scans $u_{l}.E$ to locate the first run $u_l.E[z]$ of $c$ such that the cumulative run length is equal to or greater than the local $i$. If such a run exists, it returns its global index in $L$. If the run does not exist, $c$ does not occur in suffix $L^{u_{l}}[i..]$, so the algorithm follows the leftmost branch of $y_1$ as described before.

\paragraph{Finding the branching node} To avoid backtracking $path(i)$ to find $u_k$, during descent, we maintain, for each internal node $u_{j} \in path(i)$ with $j\geq2$, a pointer $p_k$ to a candidate redirection node $y_1$ together with the compacted value $c_{k}$ of $c$ in $alphacomp(L^{u_{k}})$. 

For a root child $u_2$, $p_k$ is $u_2.Z[c]$ and $c_{k}=c$ (or null if no valid successor exists). For a general internal node $u_j \neq u_2$ where the child index for $u_{j+1}$ is $q$, $p_k$ points to the next right-hand sibling $succ_{1}(u_{j}.N_c, q+1)$ containing $c$ and $c_{k}$ is the compacted $c$ in $u_{j}$. If no valid pointer $p_k$ is found (i.e., $u_{j+1}$ does not have a right-hand sibling for $c$), the pair $(p_k, c_{k})$ is inherited from $u_{j-1}$.

Once the stopping node $u_{j}$ is found, the algorithm moves to the node $y_1$ indicated by $p_k$, updates $c=c_k$, and starts the descent of the leftmost branch for $c$ in $y_1$. On the other hand, if $p_k$ is a null pointer, $succ_c(L, i)$ returns a null value.

Algorithm~\ref{algo:succ_q} shows the successor mechanism in more detail.

%\begin{algorithm}[t]
%\DontPrintSemicolon
%\caption{$succ_c(L, i)$}\label{algo:succ_q}
%
%$u_2 \gets \text{ root child in } path(i) $\;
%
%\If{$u_2.A[c]=0$}{
%$y_1 \gets \text{node indicated by } u_2.Z[c]$\;
%\For{$y \in y_1, y_2, \ldots, y_l$}{\label{code:descend_start}
%    $c \gets rank_1(y.A, c)$\;
%}
%\Return{$\text{global idx of the first } c \text{ in } y_l.E[z]$}\label{code:descend_end}
%}
%
%$c_k \gets c, p_k \gets u_2.Z[c]\text{ }$\Comment{candidate for divergence}
%
%\For{$u_j \in u_2, u_3, \ldots, u_l \subset path(i)$}{
%
%    $q \gets \text{ child index of } u_{j+1} \text{ in } u_j$\;
%    $q' \gets succ_1(u_{j}.N_c, q+1)\text{ }$\Comment{next sib. of $u_j$ with $c$}
%    \lIf{$q' \neq \textsc{NULL}$ \text{\textbf{then}}}{ $c_k \gets c$, $p_k \gets u_j.P[q']$ }
%    
%    \If{$u_j.N_c[q]=0$}{
%        $c \gets c_k$, $y_1 \gets \text{ node indicated by } p_k$\;
%        execute Lines~\ref{code:descend_start}-\ref{code:descend_end}\;
%    }
%    $c \gets rank_1(u_j.A, c)$\;
%}
%
%$\text{find the first run } u_l.E[z] \text{ for } c \text{ after local } i \text{ in } u_l$\;
%
%\lIf{$u_l.E[z]$ \text{is found \textbf{then}}}{\Return{global idx its first $c$}}
%
%$c \gets c_k$, $y \gets \text{node indicated by } p_k$\;
%execute Lines~\ref{code:descend_start}-\ref{code:descend_end}\;
%\end{algorithm}

\paragraph{A modified successor query} The operation $backwardsearch(P)=(sp_1, ep_1, \sa[sp_1])$ internally uses a variant of $succ_c(L, i)=h_{\rho}$ that returns the index $\rho$ of the run $L[h_{\rho}..t_{\rho}]$ where $h_{\rho}$ lies. This index indicates the suffix sample $\sa_{h}[\rho]$ from which the toehold $\sa[sp_1]$ is derived (Equation~\ref{eq:toehold}). Our VLB-tree design will distribute $\sa_h$ in the leaves to improve cache efficiency (Section~\ref{sec:cnt_toehold}) so an index $\rho$ does not necessarily tell where $\sa_h[\rho]$ is stored. We define a more convenient query:  

\begin{itemize}
    \item $msucc_c(L, i)$: Let $succ_c(L, i)=i'$ be a successor index for $c$ and let $L[h_{\rho}..t_{\rho}]$ be the run where $i'$ falls in $L$. The \emph{memory successor} query $msucc_c(L, i)$ returns the byte offset within the VLB-tree where $\sa_h[\rho]$ resides. 
\end{itemize}

\paragraph{Theoretical cost} We note that the theoretical cost of the rank and successor algorithms remains the same as that of $access(i)$ (Section~\ref{sec:acc}), and they incur $O(\log_f (\ell/w))$ additional cache misses if they have to diverge the traversal. 

\section{Counting with toehold}\label{sec:cnt_toehold}

The current VLB-tree encoding (Section~\ref{sec:aug_enc}) already supports the query $rank_c(L, i)$, which is sufficient to compute the backward-search intervals $(sp_l, ep_l)$ using Equation~\ref{eq:bs}. The next step is to use Equation~\ref{eq:toehold} to obtain the toehold $\sa_h[sp_1]$. We need:

\begin{itemize}
    \item To determine the step $g$ and its position $sp_g$.
    \item To retrieve $\sa_h[\rho]$ given $g$ and $sp_g$. 
\end{itemize}

We describe how to support both operations in the VLB-tree.

\subsection{The step for the toehold} 

Cobas et al.~\cite{srindex2025} present a simple method to obtain $g$ and $sp_g$: initially, set $g=m+1$ and $sp_g=1$, and run the backward search. Whenever a step $(sp_l, ep_l)$ satisfies $L[sp_l]\neq P[l-1]$, update $g=l$ and $sp_g=sp_l$. The final values are obtained once $(sp_1, ep_1)$ is reached. 

The challenge is to test $L[sp_l] \neq P[l-1]$ in a cache-friendly manner because Equation~\ref{eq:bs} does not necessarily visit the area for $L[sp_l]$. We define a new query $\mathit{headrank}_{c}(L, i) = (o, b)$ in the VLB-tree, where $o$ is the number of occurrences of $c$ in $L[1..i]$ and $b$ is a bit set to $1$ if $L[i] = c$.

This operation behaves like $\mathit{rank}_c(L, i)$ (Section~\ref{sec:rank_q}), with one additional check:

\begin{itemize}
    \item If $path(i)$ diverges from $path(pred_c(L, i))$, then $L[i] \neq c$ and we set $b = 0$.
    \item Otherwise, both paths coincide. Let $v$ be the leaf reached and let $i'$ be the local position of $i$ in $L^{v}$. We set $b=1$ if $L^{v}[i']$ matches $c$, and $b=0$ otherwise.
\end{itemize}

A subtle issue arises because Equation~\ref{eq:bs} infers $sp_{l-1} = C[P[l-1]] + rank_{P[l-1]}(L, sp_{l}-1)+1$, whereas we now perform rank at position $sp_l$. Since $L[sp_l-1]$ and $L[sp_l]$ may lie in different fragments of the VLB-tree, we adjust the computation as follows.

We first run $headrank_{P[l-1]}(L, sp_l) = (o, b)$. If $b = 1$, we decrement $o$ by one; otherwise, $o$ is left unchanged. The value $sp_{l-1} = C[P[l-1]] + o+1$ is therefore identical to the original formulation.

\subsection{Suffix array samples}~\label{sec:vlbt-sa-samp}
After computing $g$ and $sp_g$, we must retrieve the index $\rho$ such that $succ_{P[g-1]}(L, sp_g)=h_{\rho}$. The VLB-tree query $succ_{P[l-1]}(L, sp_g)$ (Section~\ref{sec:succ_q}) returns the index $h_{\rho}$, but this alone does not identify where $\sa_h[\rho]$ is stored.

We address this by augmenting the leaves of the VLB-tree with suffix array samples. For a leaf $v$ whose sequence $v.E$ contains $m'$ runs, we store an array $W[1..m']$ such that $W[z]$ equals the value of $\sa_h$ associated with the $z$-th run of $v.E$ (from left to right).

With this structure in place, retrieving $\sa_h[\rho]$ is straightforward. Given $g$ and $sp_g$, we call $\mathit{msucc}_{P[g-1]}(L, sp_g)$ (Section~\ref{sec:aug_enc}) to obtain the leaf $v$ containing the BWT position $h_{\rho}$. Let $v.E[z]$ be the run where $h_{\rho}$ lies; the desired suffix array sample is $v.W[z] = \sa_h[\rho]$.

A key advantage of this design is that $msucc_{P[g-1]}(L, sp_g)$ and the subsequent retrieval of $\sa_h[w]$ are performed in the same memory area (the leaf $v$), ensuring excellent locality. In contrast, traditional $r$-index implementations store $\sa_h$ and $L$ in different memory locations.

\paragraph{Handling run splitting} Because the tree construction may split BWT runs across leaves, the first and last elements of $v.E$ may be \emph{partial}. In particular, $v.E[1]$ is partial if the symbol preceding $L^{v}$ in $L$ is the same as the symbol of $v.E[1]$. Each leaf stores a flag indicating whether this situation occurs.

When $v.E[1]$ is partial, its corresponding BWT area is a proper suffix $L[t_{\rho}-q+1..t_{\rho}]$ of a run $L[h_{\rho}..t_{\rho}]$. Since the values of $\sa_h$ are associated with run heads, and $v.E[1]$ does not include $L[h_{\rho}]$, the array $v.W$ does not have a value for $v.E[1]$. In this case, each entry $v.W[z]$ corresponds to $v.E[z+1]$.

This situation does not affect the retrieval of $\sa_h[\rho]$. The positions $sp_g$ and $h_{\rho}$ always belong to different BWT runs. Consequently, if $h_{\rho}$ lies in the run $v.E[z]$ of some leaf $v$, then $v.E[z]$ cannot be partial and $v.W$ is guaranteed to have a sample for $h_{\rho}$.

\section{Reducing the space usage of the VLB-tree}\label{sec:samp}

A drawback of the VLB-tree that we have not yet addressed is the space overhead introduced by the satellite data stored at the root level. Each child $u_2$  of the root stores an array $Z$, which allows the traversal of $path(i)=u_1, u_2, \ldots, u_l$ to jump to an appropriate right-hand sibling whenever a divergence is required. This information is used when answering both rank and successor queries (Section~\ref{sec:aug_enc}). Since the root may have up to $n/\ell$ children and each $Z$ array has $\sigma$ entries, the space overhead can reach $\sigma(n/\ell)$ bits in the worst case.

In this section, we show how to reduce this cost by sampling each array $Z$ while still ensuring that the operation $backwardsearch(P[1..m])$ works correctly whenever $P$ occurs in $S$. When $P$ is not in $S$, operations $rank$, $headrank$, and $msucc$ may fail, returning zero or null when they should return a non-zero or non-null value. This behavior is acceptable, since this function naturally returns an empty range ($sp_{1}>ep_{1}$) when $P$ does not appear in $S$.

\paragraph{Intuition} Each array $Z$ does not need to store entries for \emph{all} symbols in $\Sigma$. Consider the suffix array range $\sa[sp_l..ep_l]$  representing the suffixes of $S$ prefixed by $P[l..m]$. The next backward-search step computes $rank_{P[l-1]}(L, sp_l)$ (or $headrank$ in our case) to obtain $sp_{l-1}$. However, if the pattern $P[l-1..m]$ does not exist in $S$, there is no need to store an entry for $P[l-1]$ in the array $Z$ of the root child $u_2$ containing position $sp_l$. We now formalise this idea.

Let $\mathcal{S} \subset \{(s,e) \in [n]\times[n] \mid s \le e\}$ be the set of valid ranges\footnote{For readers familiar with suffix trees, $\mathcal{S}$ contains the ranges associated with the nodes of the suffix tree of $S$.} in $\sa$ that $backwardsearch(P)$ may visit when computing $(sp_1,ep_1)\in \mathcal{S}$ for any pattern $P$ in $S$.

For a pair $(p,q)$ with $1 \le p \le q \le n$, define

\begin{equation*}
\begin{array}{l}
\mathcal{L}(p,q) = \{ (s,e) \in \mathcal{S} \mid s \le p \le e \le q \} \text{ and } \\
\mathcal{R}(p,q) = \{ (s,e) \in \mathcal{S} \mid p \le s \le q \le e \},
\end{array}
\end{equation*}

representing the sets of suffix array ranges overlapping the left and right ends of $(p,q)$, respectively.

The \emph{leftmost} and \emph{rightmost} overlaps of $(p,q)$ are 

\begin{equation*}
\begin{array}{l}
lmo(p,q) = \arg\min_{(s,e)\in \mathcal{L}(p,q)} s \text{ and } \\
rmo(p,q) = \arg\max_{(s,e)\in \mathcal{R}(p,q)} e,
\end{array}
\end{equation*}

respectively, breaking ties for $lmo(p, q)$ by choosing the pair with the largest $e$, and for $rmo(p,q)$, the pair with the smallest $s$. 

For a BWT substring $L[p..q]$ and its overlapping suffix array ranges $(a,b)=lmo(p,q)$ and $(y,z)=rmo(p,q)$, let $\Sigma^{(a,p-1)} \subseteq \Sigma$ and $\Sigma^{(q+1,z)} \subseteq \Sigma$ be the alphabets of $L[a..p-1]$ and $L[q+1..z]$, respectively.  

\begin{lemma}\label{lem:samp}
Let $u_2$ be a root child in the VLB-tree with fragment $L^{u_2}=L[p..q]\in \Sigma^{*}$. The array $u_2.Z$ can be sampled so that it stores entries only for symbols in $\Sigma^{(a,p-1)} \cup \Sigma^{(q+1,z)}$, and $backwardsearch(P)$ still returns the correct result whenever $P$ occurs in $S$.
\end{lemma}

\begin{proof}
Let $l \in \{m+1, m,\ldots, 2\}$ be a step in the operation $backwardsearch(P)$, and let $\sa[sp_{l}..ep_{l}]$ be the range of suffixes in $S$ prefixed by $P[l..m]$. We say that the triple $(sp_l, ep_l, P[l-1])$ is \emph{valid} if $P[l-1]$ occurs in $L[sp_l..ep_l]$, which is equivalent to $P[l-1..m]$ occurring in $S$.

To prove the lemma, it suffices to show that whenever a rank or successor operation in $backwardsearch(P)$ process a valid index $i \in [p..q]$ and a valid symbol $c$, the sampled array $u_2.Z$ contains an entry for $c$. Indeed, a valid triple $(sp_l, ep_l, P[l-1])$ produces a range $(sp_{l-1}, ep_{l-1})$ for a suffix $P[l-1..m]$ that exists in $S$.%; an incorrect backward-search step would therefore lead to an incorrect final result when $P$ occurs in $S$.  

We focus on $rank_c(L, i)$ and $headrank_c(L, i)$; the argument for the successor queries is analogous. 

For an index $i \in [p, q]$, the query $rank_c(L, i)$ traverses $path(i)=u_1,u_2,\ldots, u_v$ and may diverge if $path(pred_c(L, i))\neq path(i)$. Furthermore, when $L^{u_2}=L[p..q]$ does not have occurrences of $c$, the algorithm falls into a special root-child case and redirects the traversal from $u_2$ to the right-hand sibling $y_1$ containing $L[succ_c(L, i)]$. The pointer to $y_1$ is stored in $u_2.Z[c]$ (Section~\ref{sec:rank_q}).

Two cases must be considered when sampling $u_2.Z$.

Case one:
\begin{equation*}
    i=sp_l \text{ and } p \leq sp_l \leq q.
\end{equation*}

In the VLB-tree, $backwardsearch(P)$ uses $headrank_{P[l-1]}(L, sp_l)$ to compute $sp_{l-1}$. It may happen that $P[l-1]$ does not occur in $L[p..q]$ but does occur in $L[q+1..ep_l]$, forcing the use of $u_2.Z[P[l-1]]$. Since $(sp_{l}, ep_l, P[l-1])$ is valid, the entry $u_2.Z[P[l-1]]$ must exist in the sampling of $u_2.Z$, otherwise $backwardsearch(P)$ fails with a pattern that exists in $S$. The sampling of $P[l-1] \in \Sigma^{(q+1,z)}$ is guaranteed because $(sp_l, ep_l) \in \mathcal{R}$ overlaps the right end of $(p,q)$, and $(y,z) = rmo(p, q)$ contains all elements of $\mathcal{R}$. 

Case two:

\begin{equation*}
    i=ep_l \text{ and } p \leq ep_l \leq q
\end{equation*}

In this case, $backwardsearch(P)$ uses $rank_{P[l-1]}(L, ep_l)$ to compute $ep_{l-1}$. The symbol $P[l-1]$ may occur in $L[s_l..p-1]$ but not in $L[p..q]$, again forcing the use of $u_2.Z[P[l-1]]$. This time, the existence of $P[l-1]$ in the sample of $u_2.Z$ is guaranteed because $(sp_l, ep_l) \in \mathcal{L}$ overlaps the left end of $L[p..q]$, and $(a,b)=lmo(p,q)$ contains all elements of $\mathcal{L}$.
 
In conclusion, a sample of $u_2.Z$ that stores information for $\Sigma^{(a,b)} \cup \Sigma^{(y,z)}$ is sufficient to ensure that $backwardsearch(P)$ is answered correctly whenever $P$ exists in $S$. Figure~\ref{fig:samp} and Example~\ref{exp:samp} illustrates both cases and the corresponding overlaps within $L[p..q]$.
\end{proof}

\begin{figure*}[t]
\centering
%\tikzsetnextfilename{root-children-samp}
\includegraphics[width=\textwidth]{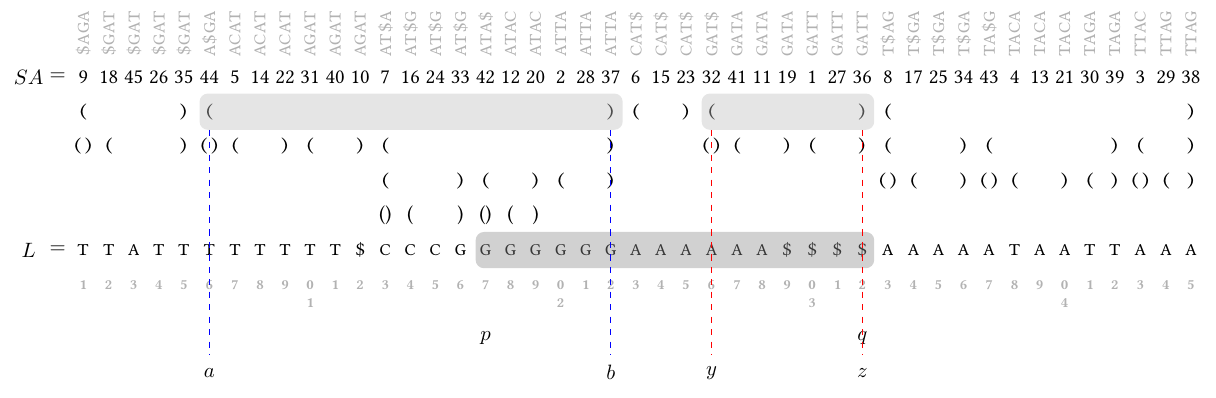}
\caption{Sampling of $Z$ arrays. The parentheses are the ranges $\sa[sp_l..ep_l]$ that $backwardsearch(P)$ can visit for patterns of length $\leq 4$. The gray block $L[p,q]=L[17..32]$ marks the fragment of a root child (the one in the middle in Figure~\ref{fig:simp_vlbt_enc}). The gray parentheses to the left highlight the longest left overlap $(a,b)=lmo(p,q)$, and the gray parentheses to the right highlight the longest right overlap $(y,z)=rmo(p,q)$. The alphabet $\Sigma^{(a,p-1)}$ is $\{\texttt{\$}, \texttt{C}, \texttt{G}, \texttt{T}\}$ and the alphabet $\Sigma^{(q+1,z)}$ is $\emptyset$. The $Z$ array for the root child encoding $L[p..q]$ thus stores pointers for $\{\texttt{\$}, \texttt{C}, \texttt{G}, \texttt{T}\}$.}
\label{fig:samp}
\Description{Sampling strategy}
\end{figure*}

\begin{example}\label{exp:samp}
Performing rank with the sampling of Figure~\ref{fig:samp}. Consider $rank_c(L, ep_l)$ for $c \in \{\texttt{\$},\texttt{C},\texttt{T}\}$, where $ep_l=b$. Since $c$ does not occur in $L[p..b]$ (and $p \le b \le q$), the VLB-tree uses the $Z$ array of the root child encoding $L[p..q]$. By Lemma~\ref{lem:samp}, $c$ has a pointer in $Z$, so the query returns a valid answer. Note that $rank_{\texttt{A}}(L,b)$ returns $0$ although the correct value is $1$ (due to $L[3]=\texttt{A}$). This inconsistency is harmless because $\texttt{A}$ does not occur in $L[sp_l=a..ep_l=b]$, so the corresponding backward-search step correctly yields an empty range ($sp_{l-1}>ep_{l-1}$).
\end{example}

\section{Locating from toehold}\label{sec:phi}

In this section, we introduce a new representation for the function $\phi^{-1}$ used by $locate(P)=\sa[sp_1..ep_1]$ to decode the suffix array values in $\sa[sp_1+1..ep_1]$ from the toehold $\sa[sp_1]$ (Section~\ref{sec:r_index}). This new encoding balances speed and space using adaptive techniques based on the VLB-tree. 

The classical data structure for $\phi^{-1}$ uses a bitvector $B[1..n]$ marking text positions that correspond to run tails in the BWT, together with an array $D[1..r]$ storing suffix array differences at run boundaries. The operation $\phi^{-1}(\sa[j])=\sa[j+1]$ is then solved using Equation~\ref{eq:phi}.

This formulation reveals two opportunities for improvement: 

\begin{itemize}
    \item The distribution of $1$s in $B$ is typically skewed, so adaptive structures improve space usage. 
    \item Interleaving $B, D$ improves spatial locality when answering $\phi^{-1}(\sa[j])$.
\end{itemize}

We combine these observations by introducing a VLB-tree for $\phi^{-1}$, denoted $\mathcal{T}_{\phi}$.

\subsection{The VLB-tree for locating}

The tree $\mathcal{T}_\phi$ is built over a logical sequence $Q[1..n]$ that interleaves the information of $B$ and $D$. Let $b_1,b_2,\ldots,b_r$ be the positions in $B$ such that each $B[b_a]=1$, and let $d_a=b_{a+1}-b_{a}$ denote the distance between consecutive $1$s (with $d_r=n+1 - b_{r}$). We define the run-length encoding of $Q$ as

\begin{equation*}
rle(Q) = (D[1],d_1)(D[2],d_2),\ldots, (D[r], d_r),
\end{equation*}

where each pair $(D[a], d_a)$ forms a run, analogous to those in the BWT. The entry in $D[a]$ is the run symbol and $d_a$ the run length.

We construct $\mathcal{T}_{\phi}$ from $Q$ following the same principles as in Section~\ref{sec:basic_end}, but storing only the information required to answer the operation $access(Q, \sa[j])$. Specifically, each internal node $u$ stores only the bitvector $u.X$ marking the superblocks and the array $u.P$ of the child pointers, while each leaf only stores the run-length sequence $E$ corresponding to a substring of $Q$.

Supporting $access(Q, \sa[j])=D[rank_1(B, \sa[j])]$ is sufficient to answer

\begin{equation*}
\phi^{-1}(\sa[j])=\sa[j]+access(Q,\sa[j]).
\end{equation*}

In the description of $\phi^{-1}(\sa[j])$ in Section~\ref{sec:r_index}, positions $\sa[j]$ that fall within $B[b_{a}..b_{a+1}-1]$ access $D[a]$. This range is represented in $\mathcal{T}_{\phi}$ by the run $(D[a],d_a)$, thereby improving spatial locality compared to traditional $r$-index implementations. 

\paragraph{Handling run splitting} The construction of $\mathcal{T}_{\phi}$ may divide a run $(D[a], d_a)$ into several parts $(D[a], d_{a.1}),(D[a], d_{a.2}), \ldots$ depending on the block size $\ell$ and the run length $d_a$. However, this partition does not affect the output of $\phi^{-1}$ because all positions $B[b_{a.2}], B[b_{a.3}],\ldots$ are associated with $D[a]$. This property follows directly the definition of the BWT (see Section~\ref{sec:r_index}).

Figure~\ref{fig:vlbtree-phi} illustrates the construction of $\mathcal{T}_{\phi}$.

\subsubsection{Complexity and adaptivity}

The traversal of $path(\sa[j])$ in $\mathcal{T}_{\phi}$ allows us to find the run $(D[a], d_a)$ containing position $\sa[j]$, which we then use to answer $access(Q, \sa[j])$. This operation takes $O(\log_f (\ell/w) + w)$ time, $\ell$, $w$, and $f$ being the VLB-tree parameters.

Areas of $Q$ containing long runs (respectively, sparse regions of $B$) generate leaves in $\mathcal{T}_{\phi}$ that are close to the root, resulting in fewer cache misses and short leaf scans. Conversely, regions with many short runs (equivalently, dense regions of $B$) are placed deeper in $\mathcal{T}_{\phi}$, increasing space usage, but matching the inherently higher cost of navigating many runs. This adaptive behavior mirrors that of the VLB-tree used for the BWT and achieves a balanced trade-off between space and query performance.

\section{Fast subsampled index}\label{sec:sr-vlbtree}

The VLB-trees for the BWT and for $\phi^{-1}$ together form a compressed suffix array supporting the queries $count(P)$ and $locate(P)$. However, storing the $2r$ suffix array samples required by these structures can be expensive for repetitive texts with many small variations.

The $sr$-index solves the space issue, but existing implementations suffer from poor cache locality because queries must traverse multiple disconnected structures (Section~\ref{sec:sr_index}). In this section, we show how to encode both the $sr$-index and its fast variant (Section~\ref{sec:sr-index-fast}) within the VLB-tree framework, improving locality while preserving their space.

\subsection{Modifications on the VLB-tree of the BWT}

Let $\mathcal{T}_L$ denote the VLB-tree of the BWT. Each leaf $v$ in $\mathcal{T}_{L}$ stores a run-length encoded fragment $v.E[1..m']$ of $L$ and an array $v.W[1..m']$ with the suffix array samples in $\sa_h$ associated with $v.E$ (Section~\ref{sec:vlbt-sa-samp}). After subsampling, $v.W$ shrinks if some of the runs in $v.E$ do not retain a sample.

To represent this situation, we augment each leaf with a bitvector $v.R[1..m']$, where $v.R[z]=0$ indicates that $v.E[z]$ retains its sample and $v.R[z]=1$ otherwise. For a run $v.E[z]$ with $v.R[z]=0$, its sample is stored at $v.W[rank_0(v.R,z)]$.

Since $m' \le m$ and $v.R$ fits in a machine word, $rank_0$ can be computed in constant time using bitwise operations.

\paragraph{Backward search} The only query in $\mathcal{T}_L$ that requires a modification is $backwardsearch(P)=(sp_1, ep_1, \sa[sp_1])$. The run $\rho$ in $L$ used to derive the toehold $\sa_h[sp_1]$ (Equations~\ref{eq:toehold}) may not retain a subsample, in which case we need to recover it (Equation~\ref{eq:sr_index_th}). We operate as follows: after completing the backward search and obtaining $(g, sp_g)$, we use $msucc_{P[g-1]}(L, sp_g)$ to access the VLB-tree leaf $v$ encoding the run $L[h_{\rho}..t_{\rho}]$ associated with $\sa_h[\rho]$. Let $v.E[z]$ be run for $L[h_{\rho}..t_{\rho}]$. If $v.R[z]=0$, we return $v.W[rank_0(v.R', z)]=\sa_h[\rho]$, otherwise we recover $\sa_h[\rho]$ via $LF$ starting from $h_{\rho}$, as in Equation~\ref{eq:sr_index_th}. 

\subsection{Modifications on the VLB-tree of locating}

In the standard $sr$-index, the query $\phi^{-1}(\sa[j])$ performs up to $s$ $LF$ steps to check if $\sa[j]$ falls in valid area. A failed test indicates invalid area but returns $\sa[j+1]$ altogether, whereas a successful test does not give the answer but indicates that Equation~\ref{eq:phi} with subsampled $B$ and $D$ returns the correct $\sa[j+1]$ (Section~\ref{sec:sr_index}).

This functionality does not require any modification to the VLB-tree $\mathcal{T}_{\phi}$. We can build it directly over the logical sequence $Q$ derived from the subsampled $B$ and $D$.

On the other hand, the fast variant of the $sr$-index marks the invalid areas to save computations (Section~\ref{sec:sr-index-fast}) and therefore we must incorporate this information into the VLB-tree $\mathcal{T}_{\phi}$.

A run $(D[a], d_a)$ in $Q$ represents a range $B[b_a..b_{a-1}-1]$ that is fully or partially valid. The run $(D[a], d_a)$ of a partially valid range can be split into two halves $d_{a,l}, d_{a,r}$. Let $b_a < b' < b_{a+1}$ be the leftmost cleared position within $B[b_a..b_{a+1}-1]$ during subsampling. The left half $d_{a,l}$ refers to the valid prefix $B[b_a..b'-1]$ while $d_{a,r}$ refers to the invalid suffix $B[b'..b_{a+1}-1]$. In the VLB-tree of $\mathcal{T}_{\phi}$, a query $access(Q, \sa[j])$ such that $\sa[j]$ falls within $d_{a,l}$ must return $D[a]$, while a query that falls in $d_{a,r}$ must return null to indicate an invalid area. A query within a fully valid area always returns the corresponding difference value in $D$.

\begin{figure}[t]
\centering
%\tikzsetnextfilename{vlb-tree-phi}
\includegraphics[width=0.95\linewidth]{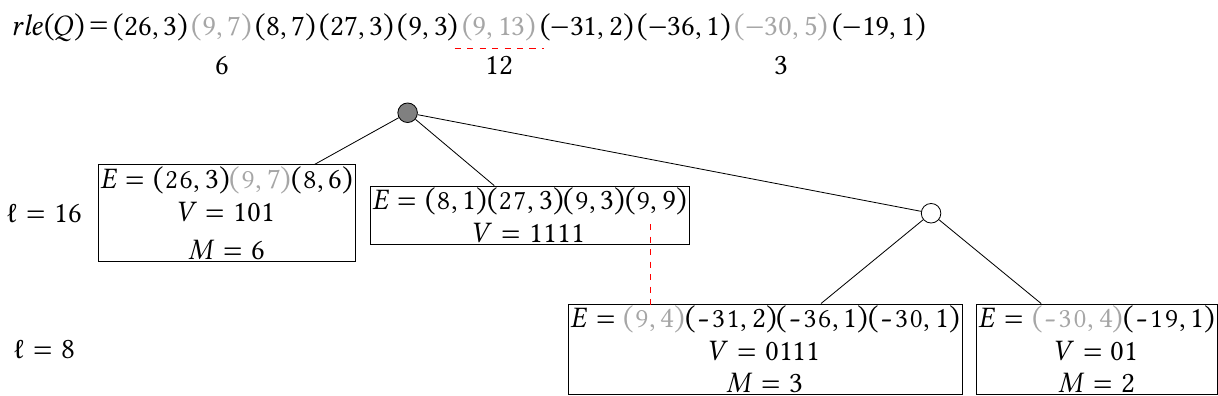}
\caption{VLB-tree $\mathcal{T}_{\phi}$ of the sequence $Q$ generated from subsampled $B$ and $D$ of Figure~\ref{fig:sri}. The parameters are $\ell=16,f=2$ and $w=4$. Gray elements in $rle(Q)$ are partially valid areas. The number below is the length of the prefix that remains valid (gray regions in $B$ of Figure~\ref{fig:sri}). Bitvector $V$ and array $M$ are elements of the fast $sr$-index (Section~\ref{sec:sr-vlbtree}). Dashed lines correspond to Example~\ref{ex:run_split}.} 
\label{fig:vlbtree-phi}
\Description{VLB-tree for $\phi^{-1}$}
\end{figure}

We incorporate this information into $\mathcal{T}_{\phi}$ using analogous structures to those of the fast $sr$-index, augmenting each leaf with:

\begin{itemize}
    \item A bitvector $v.V[1..m']$, where $v.V[z]$ indicates the partial or full validity of the run $v.E[z]$, and
    \item An array $v.M[1..rank_0(v.V, m')]$ that stores in $v.M[rank_0(v.V, z)]$ the valid prefix $d_l$ of a partially valid run $v.E[z]$.
\end{itemize}

\paragraph{Handling run splitting} The only relevant consideration when building $\mathcal{T}_{\phi}$ is the update of $d_{a,l}$ when a partially valid run $(D[a], d_a)$ is split into multiple VLB-tree leaves. Suppose the construction algorithm partitions $(D[a], d_a)$ into $p$ fragments $(D[a], d_{a.1}), (D[a], d_{a.2}), \ldots (D[a], d_{a.p})$. If the fragment $(D[a], d_{a.j})$ is contained within the prefix $d_{a,l}$ of $(D[a], d_a)$, then $(D[a], d_{a.j})$ becomes fully valid. On the other hand, if $(D[a], d_{a.j})$ is contained within the prefix $d_{a,r}$ of $(D[a], d_a)$, the fragment becomes fully invalid. Finally, when $(D[a], d_{a.j})$ falls within a substring that overlaps $d_l$ and $d_r$, it remains partially valid, but its valid prefix must be updated.

Figure~\ref{fig:vlbtree-phi} shows examples of $v.V$ and $v.M$, and Example~\ref{ex:run_split} illustrates the splitting of runs.

\begin{example}\label{ex:run_split}
     In Figure~\ref{fig:vlbtree-phi}, the run $rle(Q)[6]=(\texttt{9}, 13)$ (red dashed underline) corresponds to $Q[23..35]$ (equivalently, $B[23..35]$), with associated difference value $\texttt{9}$ in $D$. The run is partially valid, with valid prefix length $12$. During the construction of $\mathcal{T}_{\phi}$, the run is divided into two parts, $(\texttt{9}, 9)$ and $(\texttt{9},4)$ (vertical dashed red line). The fragment $(\texttt{9},9)$ becomes fully valid because it lies within the valid prefix, while $(\texttt{9}, 4)$ remains partially valid and its valid prefix becomes $12-9=3$. Accessing $Q[23..34]$ in $\mathcal{T}_{\phi}$ returns $\texttt{9}$, whereas accessing $Q[35]$ (the invalid suffix) returns a null value.
\end{example}

\section{Experiments}

We implemented our data structures in a \texttt{C++} library called VLBT (\url{https://github.com/ddiazdom/VLBT}). The VLB-tree representation of the BWT (Section~\ref{sec:aug_enc}) is denoted \vlbtbwt and supports $count$ queries. The compressed suffix array that combines the BWT and $\phi^{-1}$ VLB-trees with fast subsampling (Section~\ref{sec:sr-vlbtree}) is denoted \vlbtsri and supports both $count$ and $locate$ queries. Both implementations rely on SIMD instructions to accelerate scans over run-length encoded sequences.

\subsection{Experimental setup}

We evaluate the tradeoffs between query speed, cache efficiency, and index size for BWT-based data structures supporting $count$ and $locate$ queries.

We fixed the VLBT parameters $w=64$ and $f=4$ to implement the operations $rank_1, select_1$ and $pred_1$ using bitwise operations on 64-bit machine words. Then, we built multiple variants of \vlbtbwt and \vlbtsri by varying the block size $\ell$ from $4^{6}=4{,}096$ to $4^{9}=262{,}144$. In all \vlbtsri instances, the VLB-trees of $L$ and $\phi^{-1}$ used the same block size $\ell$, and for each block size, we varied the subsampling rate $s$ across values $4, 8, 16, 32$. A variant named \vlbtsri-\texttt{x} corresponds to $\ell=4^{x}$, where $x \in \{6,7,8,9\}$. 

\subsubsection{Competitor tools and measurements}\label{sec:tools}

We compared our VLBT variants against the following implementations:

\begin{itemize}
    \item \mn\footnote{\url{https://github.com/simongog/sdsl-lite/blob/master/include/sdsl/wt_rlmn.hpp} (release 2.1.1)}: the run-length BWT of M\"akinen and Navarro~\cite{makinen2005succinct} as implemented in~\cite{gog2014theory}.
    \item \gkk\footnote{\url{https://github.com/dominikkempa/faster-minuter} (commit \texttt{9238178})} : BWT encoding using fixed block boosting~\cite{gog2019fixed}.
    \item \movc and \movl\footnote{\url{https://github.com/LukasNalbach/Move-r} (commit \texttt{bed2fe9})}: optimized implementations~\cite{bertram2024move} of the move structure~\cite{move21}. The variant \movc encodes only the BWT, while \movl also includes $r$-suffix array samples.
    \item \ri\footnote{\url{https://github.com/nicolaprezza/r-index} (commit \texttt{7009b53})}: the original $r$-index~\cite{g2018op} .
    \item \sriva\footnote{\url{https://github.com/duscob/sr-index} (commit \texttt{f99b54a})}: the original $sr$-index with valid $\phi^{-1}$ areas ~\cite{srindex2025} (i.e., fast variant). We varied the sampling $s$ across values $8,12,16,20$.
\end{itemize}

We created a standard command line interface to measure the speed performance of $count$ and $locate$ queries in run-length BWTs and $r$-index implementations. We excluded \movc and \movl because its codebase is substantially more involved, and used their own interface. All interfaces receive as input an index and a set of patterns, and report the average speed for $count$ and/or $locate$. The source code for these interfaces can be found in the folder \texttt{test\_suite} in the VLBT repository.

We split the experiments according to the supported query types. For structures that support only $count$, we evaluated \vlbtbwt with \mn, \gkk, and \movc. For structures that support $count$ and $locate$, we evaluated \vlbtsri against \movl, \ri, and \sriva.

\paragraph{Comparisons for \vlbtbwt} For $count$ queries, we identified the fastest and the slowest \vlbtbwt variants and compared their query times against the other methods, reporting the resulting time range. We did the analogous comparison for space by taking the smallest and the largest \vlbtbwt variants. Section~\ref{sec:exp_count_rlbwt} summarizes these ranges.

\paragraph{Comparisons of \vlbtsri} Here, performance depends on two parameters: block size and subsampling. For each fixed block size in \sriva, we averaged query time (and space) over the tested subsampling values, and then reported the best--worst range across block sizes when comparing against other tools. Since \sriva also varies sampling, we likewise averaged its query time and space over its sampling values. When comparing \vlbtsri and \sriva directly, we used averages computed over the same subsampling values. Section~\ref{sec:exp_loc_csa} presents the results. 

The tradeoffs of query-time versus space usage of all experiments are reported in Figure~\ref{fig:count_loc_time_exp}. 

To assess cache efficiency, we measured the average number of L1 data cache misses per query symbol. For $count$ queries, we compared \vlbtbwt, \mn, and \gkk; for $count$ and $locate$ queries, we compared \vlbtsri, \ri, and \sriva. For each tool, we generated an executable that loads the index, executes queries, and records cache misses. This required special compilation steps described in Section~\ref{sec:comp}. Due to compilation constraints, we could not include \movc and \movl in the cache experiments. Results are shown in Figure~\ref{fig:cou_loc_cmiss_exp}.
 
\begin{table}[t]
    \centering
    \begin{tabular}{lcccc}
    \toprule
        Dataset & Size (GB) & Alphabet & $n/r$  & Longest run (MB)\\
    \midrule
       \bac     & 133.12 & 7  & 116.77 & 0.23 \\
       \covid   & 267.41 & 17 & 940.49 & 4.11 \\
       \hum     & 241.24 & 7  & 61.82  & 6.66 \\
       \kernel  & 54.45  & 190& 263.11 & 70.6 \\
    \bottomrule
    \end{tabular}
    \caption{Datasets.}
    \label{tab:datasets}
\end{table}

\subsubsection{Datasets, queries, and indexing}

We built all indexes on four datasets with varying degrees of repetitiveness. Table~\ref{tab:datasets} summarizes the main characteristics of the collections.

\begin{itemize}
    \item \bac: genome assemblies of 30 bacterial species from the AllTheBacteria collection~\cite{atb2024}. Strings belonging to the same species are highly repetitive, while strings from different species are dissimilar.
    \item \covid: $4{,}494{,}508$ SARS-CoV-2 genomes downloaded from the NCBI genome portal \footnote{\url{https://uud.ncbi.nlm.nih.gov/home/genomes}}. These sequences are short and near-identical, with an average length of $29{,}748$.
    \item \hum: genome assemblies of 40 individuals from the Human Pangenome Reference Consortium (HPRC)~\cite{hprc}. Individual genomes are near-identical, but assembly differences introduced variability. 
    \item \kernel: $2{,}609{,}417$ versions of the Linux kernel repository\footnote{\url{https://github.com/torvalds/linux}}. This dataset is highly repetitive and has a large alphabet. 
\end{itemize}

In DNA collections (\bac,\covid, and \hum), we also considered the DNA reverse complement of each string (as is standard in bioinformatics). The numbers presented in Table~\ref{tab:datasets} already consider these extra sequences. 

We sampled $50{,}000$ random patterns of length $105$ from each dataset for the $count$ experiments. For $locate$ queries, we restricted the evaluation to patterns with fewer than $50{,}000$ occurrences in the dataset to limit computational cost. For $count$ queries, we report the average time (in $\mu$secs) per pattern, computed as the total time over all queries divided by the number of queries. For $locate$, we report the average time (also in $\mu$secs) per reported occurrence, computed as the total time divided by the number of reported occurrences returned by the queries. 

We did not measure index construction time, as the compared tools rely on different building strategies and optimizations, making a fair comparison difficult and beyond the scope of this work. To ensure consistent query results, all indexes were built using BWTs and suffix array samples generated by the same tool, \texttt{bigBWT}~\cite{bigbwt}.

\subsubsection{Machine and compilation}\label{sec:comp}

All experiments were conducted on a machine running SUSE Linux Enterprise Server 15, equipped with 1~TB of RAM and an AMD EPYC 7713 processor at 2~GHz with 128 cores. No other relevant processes were running at the time we performed our measurements.

The tools were compiled using \texttt{gcc 13.3} with optimization flag \texttt{-O3}. Our VLBT implementations use SIMD instructions from the SSE4.2 instruction set, enabled via the \texttt{-msse4.2} compilation flag.

We measured cache misses using the Cray Programming Environment (CPE), which provides access to hardware performance counters. We compiled the executables for the cache experiments (one per competitor tool) using \texttt{Cray clang 18.0.1} and set the environment variable \texttt{PAT\_RT\_PERFCTR=PAPI\_L1\_DCM} to record L1 data cache misses during query execution. Finally, we instrumented each executable  using \texttt{pat\_build}, producing executables with the \texttt{+pat} suffix that were used to collect the measurements.

Each executable receives as input a data structure and a set of patterns. First, it warms up the data structure by executing a sequence of random queries. This step ensures that no page faults occur when performing real cache measurements. The executable then pollutes the L1 data cache with random bits to force a cold-cache state. Finally, it executes the queries in random order and records the number of misses.

\begin{figure*}[t]
    \centering
    %\tikzsetnextfilename{count_locate_plots_time}
    \includegraphics[width=\linewidth]{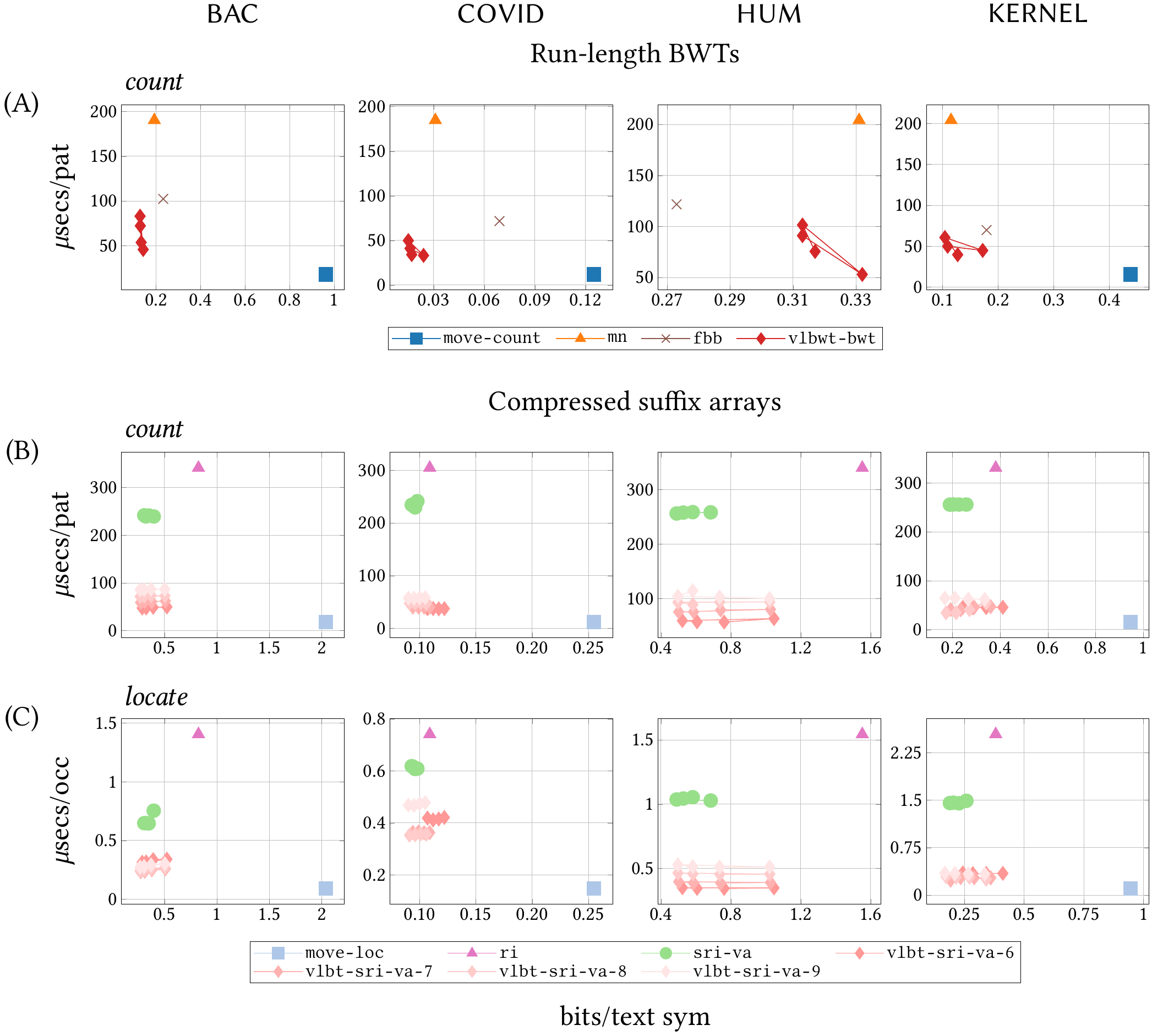}
    \caption{Query speed/space tradeoffs. The points in each \texttt{vlbt-bwt} are different block sizes $\ell$, and the points in each \texttt{vlbt-sri-va-x} are subsampling values $s$.}
    \Description{Query speed/space tradeoffs}
    \label{fig:count_loc_time_exp}
\end{figure*}

\subsection{Query speed/space tradeoffs in count queries}\label{sec:exp_count_rlbwt}

The VLBT variants supporting $count$ operations improve query time over state-of-the-art run-length BWTs and compressed suffix arrays, while remaining competitive in space. Figure~\ref{fig:count_loc_time_exp} summarizes the resulting space--time tradeoffs: panel~(A) compares BWT-only structures, and panel~(B) compares full compressed suffix arrays. Across datasets, the main outlier in speed is the move data structure, which is faster but substantially larger.

\subsubsection{Run-length BWTs}

The \vlbtbwt instances are between $1.15$ and $2.3$ times faster than \gkk, while typically being between $1.06$ and $4.6$ times smaller. The size gap is particularly pronounced on \covid, where \vlbtbwt is $2.88\text{--}4.60$ times smaller. This is consistent with the VLB-tree adapting its block boundaries to local run density, whereas \gkk uses fixed-size blocks.

On the other hand, \gkk achieves slightly better space usage on \hum, being $12.77\%\text{--}17.77\%$ smaller than \vlbtbwt. This indicates that, for this dataset, fixed-block boosting can compress some locally dense regions slightly more effectively than our run-length encoding within leaves.

When compared to \mn, \vlbtbwt outperforms it in both space and time, achieving speedups of $2.01\text{--}5.55$ while reducing space by up to a factor of $2.07$. 

Among the competing approaches, \movc offers a clear speed--space tradeoff: it is $2.56\text{--}4.54$ times faster than \vlbtbwt but requires $2.55\text{--}8.33$ times more space. Conversely, \vlbtbwt targets the opposite regime, prioritizing compactness while still improving over \mn and \gkk in query time.

\subsubsection{Compressed suffix arrays}

For $count$ queries on compressed suffix arrays, \vlbtsri consistently improves over both \ri and \sriva in query time. In terms of space, \vlbtsri and \sriva are broadly comparable when using equal subsampling settings, while \movl remains the fastest option at the expense of substantially larger space consumption (Figure~\ref{fig:count_loc_time_exp}B).

Specifically, \movl is typically $2.63\text{--}4.76$ times faster than \vlbtsri, but also $2.23\text{--}5.63$ times larger. The gap with \ri is even more pronounced: \ri is $3.22\text{--}8.14$ times slower and up to $2.28$ times larger than \vlbtsri. 

Finally, when comparing $sr$-index variants across subsampling values $s$, \sriva is $2.37\text{--}6.81$ times slower, while using comparable space. Depending on the chosen block size $\ell$ for \vlbtsri, \sriva can be almost equal in size or slightly larger, but the differences are modest compared to the consistent query-time gap. Overall, these results match our design goal: we preserve the components of the $sr$-index while reorganizing them to improve locality and reduce per-step overhead.

\subsection{Query speed/space tradeoffs in locate queries}\label{sec:exp_loc_csa}

The $locate$ results follow the same qualitative pattern across datasets (Figure~\ref{fig:count_loc_time_exp}C). Our \vlbtsri variants improve query time over \ri and \sriva while remaining close to the space usage of subsampled indexes. Compared to \movl, the results again highlight a clear space--time tradeoff: \movl is faster, but substantially larger.

Overall, our \vlbtsri variants are between $1.57$ and $9.64$ times faster than \ri, while requiring up to $2.28$ times less space. The smallest improvement is observed on the \covid dataset, where our method achieves approximately a twofold speedup while using comparable space. This behavior is expected given the high repetitiveness of this dataset. Conversely, the largest space difference ($2.28$ factor) occurred in \bac, where sequence redundancy occurs within the same bacterial species, but not across species. 

Compared to \sriva, the average speedups of our method range from $1.3$ to $5.38$, which is still considerable, but smaller than the gains in $count$ queries when comparing the same data structures ($2.37\text{--}6.36$). A possible explanation is that, although a single operation $\phi^{-1}(\sa[j])$ is more cache-friendly than its counterpart in $\sriva$, decoding consecutive values in $\sa[sp_1+1..ep_1]$ remains cache inefficient in both cases.

The \movl tool remains faster than our \vlbtsri variants, achieving speedups of $2.38$ to $3.57$ on average. However, this performance comes at a substantial space cost, as \movl requires between $2.23$ and $5.63$ times more space. This difference is particularly pronounced for the \bac dataset, where the space overhead reaches a factor of $5.34\text{--}5.63$. In contrast, the gap is smaller for \covid, with a factor of $2.23\text{--}2.62$. These results are consistent with the observation that BWT-based approaches scale poorly in the presence of edits, as in \bac. While subsampling mitigates space growth, our adaptive encoding effectively improves query speed without incurring comparable space overhead.

\begin{figure*}[t]
    \centering
    %\tikzsetnextfilename{count_locate_plots_cmiss}
    \includegraphics[width=\linewidth]{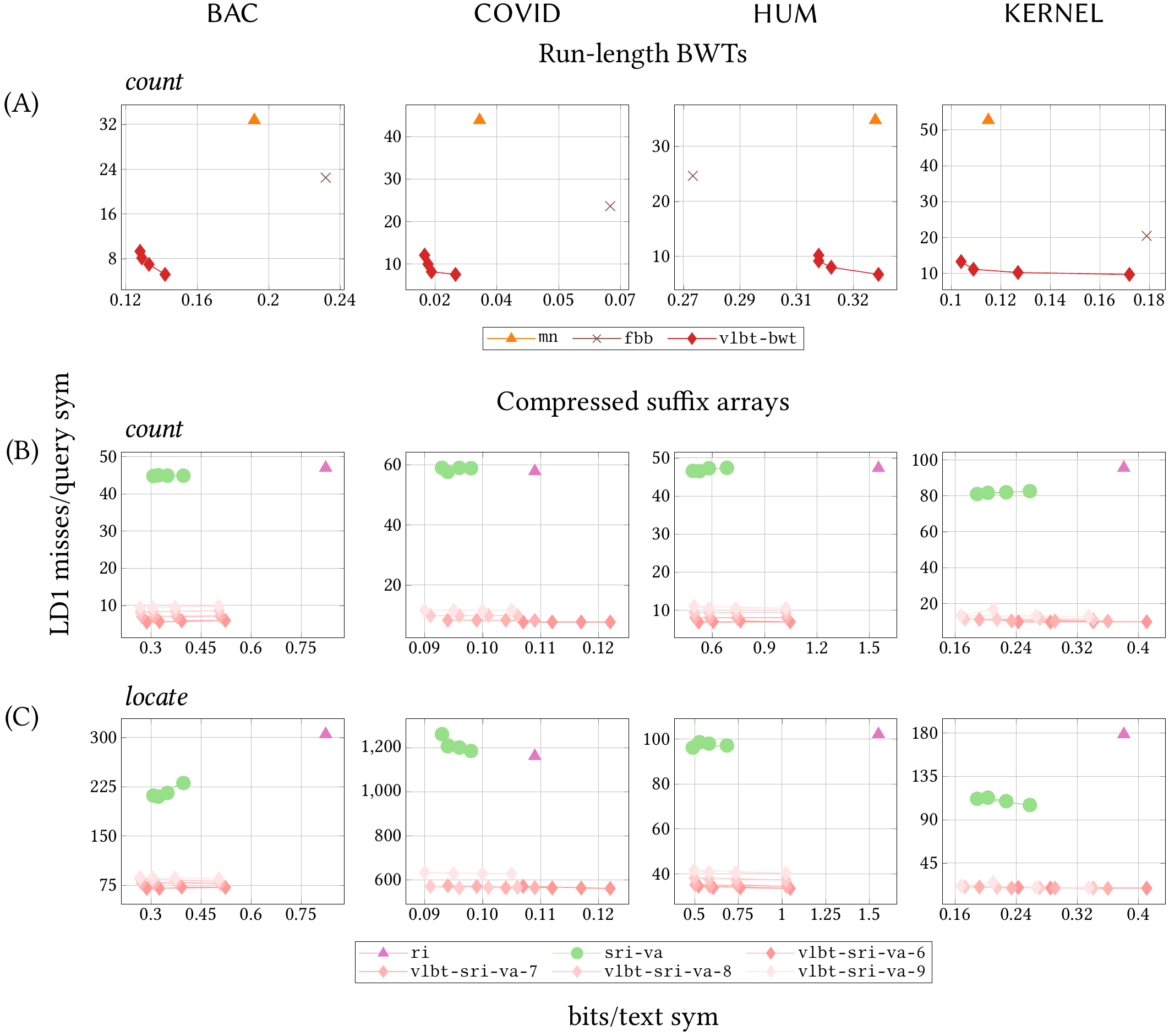}
    \caption{L1 data cache misses/space tradeoffs.}
    \Description{L1 data cache misses/space tradeoffs}
    \label{fig:cou_loc_cmiss_exp}
\end{figure*}

\subsection{Cache efficiency}

Cache-miss measurements support the main explanation for the speedups observed above (Figure~\ref{fig:cou_loc_cmiss_exp}). Across datasets, our VLBT variants incur fewer L1 data-cache misses per query symbol than \gkk, \mn, \ri, and \sriva, which is consistent with the corresponding reductions in query time (Figure~\ref{fig:count_loc_time_exp}). This indicates that memory locality is a primary driver of the observed performance differences.

For $count$ queries on run-length BWTs, \gkk incurs $1.54\text{--}4.35$ times more misses than \vlbtbwt, while \mn incurs $3.41\text{--}6.33$ times more misses. For compressed suffix arrays, the gap is larger: \ri and \sriva trigger $4.51\text{--}9.84$ and $4.48\text{--}8.46$ times more misses on average, respectively.

For $locate$ queries, \ri triggers $1.85\text{--}9.34$ times more cache misses than \vlbtsri, while \sriva incurs $1.93\text{--}5.7$ times more misses.

\begin{figure*}[t]
\centering
\begin{tikzpicture}
  \node (figA) {\includestandalone[width=0.95\linewidth]{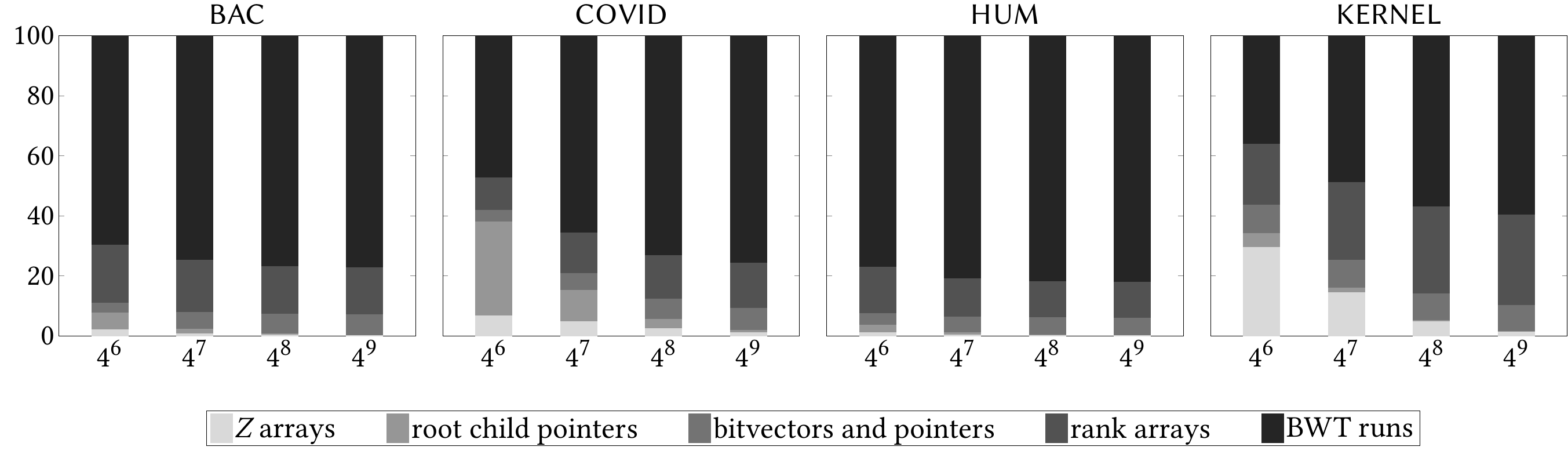}};
  \node[anchor=north west, xshift=-10pt, yshift=5pt] at (figA.north west) {\Large (A)};
\end{tikzpicture}
\vspace{1em}
\begin{tikzpicture}
  \node (figB) {\includestandalone[width=0.95\linewidth]{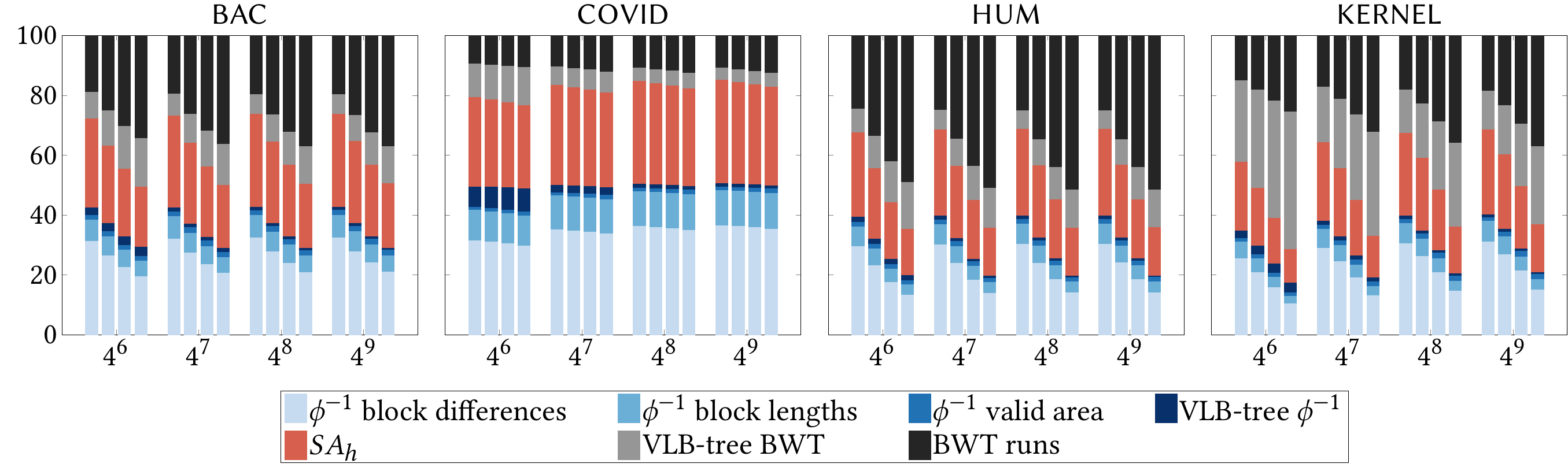}};
  \node[anchor=north west, xshift=-10pt, yshift=5pt] at (figB.north west) {\Large (B)};
\end{tikzpicture}
\caption{Space breakdown in (A) \vlbtbwt and (B) \vlbtsri. The x-axis is the block size, and different bars in the same block size represent different subsampling values.}
\label{fig:break_down}
\Description{Space breakdown for VLBT-based compresed suffix arrays}
\end{figure*}

\subsection{Effect of block size and subsampling}

In \vlbtbwt, increasing the block size from $\ell=4^6$ to $\ell=4^9$ reduces space usage by $5.72\%\text{--}39.53\%$, depending on the dataset. The largest reduction ($39.53\%$) occurs in \covid, which is consistent with larger alphabets increasing the fixed per-block overhead (e.g., alphabet-related metadata), so using fewer blocks yields larger savings. The smallest reduction ($5.72\%$) occurs in \hum, where the alphabet is small and the space is dominated by the encoded run structure rather than per-block metadata.

For \vlbtsri, changing $\ell$ has the same qualitative effect, and subsampling provides an additional (often larger) lever. The largest subsampling-driven reduction occurs in \hum: increasing $s$ from $4$ to $32$ reduces the index size by approximately $51\%$.

\subsection{Breakdown of space usage}

In \vlbtbwt, the encoded BWT runs account for the largest fraction of the total space ($36\%\text{--}82\%$). As $\ell$ increases, this fraction grows because per-block satellite data is amortized across fewer larger blocks. Alphabet-related metadata can also be visible in the breakdown: for \kernel, the root-level routing arrays $Z$ (Section~\ref{sec:aug_enc}) reach up to $29\%$ of the total space at $\ell=4^{6}$. This overhead is reduced by the sampling technique of Section~\ref{sec:samp} and by using larger block sizes (Figure~\ref{fig:break_down}A).

In \vlbtsri, suffix-array samples (stored across $\sa_h$ and the VLB-tree for $\phi^{-1}$) dominate the space budget, contributing between $21\%$ and $71.09\%$ of the total. Increasing the subsampling parameter $s$ reduces this component by $3.02\%\text{--}29.18\%$ when moving from $s=4$ to $s=32$ (smallest reduction in \covid with $\ell=4^{9}$; largest reduction in \hum with $\ell=4^{6}$). The VLB-tree overhead (BWT and $\phi^{-1}$ trees plus validity metadata) is typically a smaller fraction ($6.27\%\text{--}21.72\%$), except in \kernel where it reaches $50\%$ (Figure~\ref{fig:break_down}B).

\section{Conclusion}

We presented the VLB-tree framework, an adaptive, cache-friendly encoding for BWT-based compressed suffix arrays. Across datasets, our structures improve query time over state-of-the-art $r$-index-based baselines while remaining competitive in space with the fast subsampled variant. The move data structure remains faster, but at a substantial space cost, highlighting a practical speed--space tradeoff.

Two directions appear particularly important for further practical progress. First, improving locality for $\phi^{-1}$ decoding remains challenging: consecutive values in $\sa[sp_1..ep_1]$ are not necessarily recovered from contiguous areas of the $\phi^{-1}$ representation, which can limit $locate$ performance. Second, scalable construction is still a bottleneck for terabyte-scale deployments. While recent work has improved BWT construction from compressed text~\cite{grlbwt,lg}, comparable techniques for obtaining suffix-array samples at scale remain limited.

From a theoretical point of view, an open challenge is to derive space bounds that reflect the adaptivity of our method. In particular, worst-case guarantees for the VLB-tree of the BWT expressed only in terms of $r$, $\ell$, and $\sigma$ may fail to capture the main advantage of the approach: both $\ell$ and the satellite data adjust to local sequence content. Likewise, although the $Z$ arrays near the top of the tree introduce an apparent $\sigma$-factor overhead, in practice this cost is moderated by our sampling strategy. A probabilistic analysis of the expected block structure would likely yield bounds that better match observed behavior. 

\bibliographystyle{ACM-Reference-Format}
\bibliography{references}
%
%
%%%
%%% If your work has an appendix, this is the place to put it.
%\appendix
%Lorem ipsum dolor sit amet, consectetur adipiscing elit. Morbi
%malesuada, quam in pulvinar varius, metus nunc fermentum urna, id
%sollicitudin purus odio sit amet enim. Aliquam ullamcorper eu ipsum
%vel mollis. Curabitur quis dictum nisl. Phasellus vel semper risus, et
%lacinia dolor. Integer ultricies commodo sem nec semper.
%Etiam commodo feugiat nisl pulvinar pellentesque. Etiam auctor sodales
%ligula, non varius nibh pulvinar semper. Suspendisse nec lectus non
%ipsum convallis congue hendrerit vitae sapien. Donec at laoreet
%eros. Vivamus non purus placerat, scelerisque diam eu, cursus
%ante. Etiam aliquam tortor auctor efficitur mattis.
%Nam id fermentum dui. Suspendisse sagittis tortor a nulla mollis, in
%pulvinar ex pretium. Sed interdum orci quis metus euismod, et sagittis
%enim maximus. Vestibulum gravida massa ut felis suscipit
%congue. Quisque mattis elit a risus ultrices commodo venenatis eget
%dui. Etiam sagittis eleifend elementum.
%
%Nam interdum magna at lectus dignissim, ac dignissim lorem
%rhoncus. Maecenas eu arcu ac neque placerat aliquam. Nunc pulvinar
%massa et mattis lacinia.

\end{document}